\documentclass[11pt,a4paper]{article}
\usepackage{cite}
\usepackage{a4wide, amssymb, amsthm, amsmath, filecontents}

\makeatletter
\def\prob#1#2#3{\goodbreak\begin{list}{}{\labelwidth\z@ \itemindent-\leftmargin
                        \itemsep\z@  \topsep6\p@\@plus6\p@
                        \let\makelabel\descriptionlabel}
                \item[\it Name]#1
                \item[\it Instance]#2
                \item[\it Output]#3
                \end{list}}
\makeatother

\newcommand{\vecxny}{\vecx\wedge\vecy}
\newcommand{\vecxuy}{\vecx\vee\vecy}
\def\extR{\overline{\mathbb{R}}_{\geq 0}}
\def\posQ{\mathbb{Q}_{\geq 0}}
\def\bddQ{[0,1]_{\mathbb{Q}}}
\def\mincost{\mathrm{minCost}}
\def\dom{D}
\def\MM{\mathbf{M}}
\def\Zivny{{\v{Z}}ivn{\'y}}

\def\toValsymb{\ell}
\def\toVal#1{\toValsymb(#1)}
\def\toRel#1{\mathrm{Feas}(#1)}
\def\nfuns{\mathcal{F}}
\def\nfinfuns{\mathcal{G}}
\def\nhfuns{\mathcal{H}}
\def\vfuns{\Phi}
\def\rels{\Gamma}

\def\vclone#1{\langle#1\rangle_V}
\def\fclone#1{\langle#1\rangle_\#}
\def\nset{\mathbb{N}}
\newcommand\func{\operatorname{Func}}
\def\funs#1{\func(\dom,#1)}
\def\funsk#1{\func_k(\dom,#1)}
\def\IMP{\mathrm{IMP}}
\def\imp{\mathrm{imp}}

\def\EQ{\mathrm{EQ}}
\def\eq{\mathrm{eq}}

\newcommand\NEQ{\mathrm{NEQ}}
\newcommand\uns[1]{\mathcal{U}_{#1}}

\let\oldphi\phi
\renewcommand{\phi}{\varphi}
\renewcommand{\epsilon}{\varepsilon}

\newcommand{\veca}{\mathbf{a}}
\newcommand{\vecb}{\mathbf{b}}

\newcommand{\vecu}{\mathbf{u}}
\newcommand{\vecv}{\mathbf{v}}

\newcommand{\vecx}{\mathbf{x}}
\newcommand{\vecy}{\mathbf{y}}
\newcommand{\vecz}{\mathbf{z}}

\newcommand{\transposesymb}{\mathrm{t}}
\newcommand{\transpose}[1]{#1^{\transposesymb}}

\newcommand{\FP}{\ensuremath{\mathrm{FP}}}
\newcommand{\NP}{\ensuremath{\mathrm{NP}}}
\newcommand{\numP}{\ensuremath{\mathrm{\#P}}}

\newcommand{\LSM}{\ensuremath{\mathsf{LSM}}}

\newcommand{\nBIS}{\ensuremath{\mathrm{\#BIS}}}
\newcommand{\nSAT}{\ensuremath{\mathrm{\#SAT}}}
\newcommand{\nCSP}{\ensuremath{\mathrm{\#CSP}}}
\newcommand{\VCSP}{\ensuremath{\mathrm{VCSP}}}

\newcommand{\RHPi}{\ensuremath{\mathrm{\#RH}\Pi_1}}
\newcommand{\APred}{\leq_\mathrm{AP}}

\newtheorem{theorem}{Theorem}
\newtheorem{lemma}[theorem]{Lemma}
\newtheorem{corollary}[theorem]{Corollary}
\theoremstyle{definition}
\newtheorem{observation}[theorem]{Observation}
\newtheorem{remark}[theorem]{Remark} 
\newtheorem{definition}[theorem]{Definition}

\title {The complexity of approximating \\conservative counting CSPs\thanks{
This work was partially supported by the EPSRC Research Grant ``Computational Counting''.
}}
\author{ Xi Chen\thanks{
        Department of Computer Science, Columbia University,
        450 Computer Science Building,
        1214 Amsterdam Avenue, Mailcode:~0401, New York,
        NY~10027-7003, USA.  Supported by NSF Grant CCF-1139915 and
        start-up funds from Columbia University.}
    \and Martin Dyer\thanks{
        School of Computing, University of Leeds, Leeds, LS2~9JT, UK.}
    \and Leslie Ann Goldberg\thanks{
  Department of Computer Science, University of Oxford, Wolfson Building, Parks Road, Oxford, OX1~3QD, UK.  
 The research leading to these results has received funding from 
the European Research Council under the European Union's Seventh Framework Programme (FP7/2007-2013) ERC grant agreement no.\ 334828. The paper 
reflects only the authors' views and not the views of the ERC or the European Commission. The European Union is not liable for any use that may be made of the information contained therein.}
    \and Mark Jerrum\thanks{
        School of Mathematical Sciences, Queen Mary, University of
        London, Mile End Road, London, E1~4NS, UK.}
    \and \linebreak[4] Pinyan Lu\thanks{
        Microsoft Research Asia, Microsoft Shanghai Technology Park,
        No~999, Zixing Road, Minhang District, Shanghai, 200241, China.}
    \and Colin McQuillan\thanks{Department of Computer Science, University of Liverpool,
        Liverpool, L69~3BX, UK.}
    \and David Richerby\footnotemark[4]}
      
\begin{document}

\maketitle

\begin{abstract} 
    We study the complexity of approximately solving the weighted
    counting constraint satisfaction problem $\nCSP(\nfuns)$.  In the
    conservative case, where $\nfuns$ contains all unary functions,
    there is a classification known for the case in which the
    domain of functions in $\nfuns$ is Boolean.  In this paper, we
    give a classification for the more general problem where functions
    in~$\nfuns$ have an arbitrary finite domain.  We define the notions of
    \emph{weak log-modularity} and \emph{weak log-supermodularity}.
    We show that if $\nfuns$ is weakly log-modular, then
    $\nCSP(\nfuns)$ is in \FP{}.  Otherwise, it is at least as difficult to
    approximate as \nBIS{}, the problem of counting independent sets
    in bipartite graphs. \nBIS{} is complete with respect to
    approximation-preserving reductions for a logically defined
    complexity class \RHPi{}, and is believed to be intractable.  We
    further sub-divide the \nBIS{}-hard case.  If $\nfuns$ is weakly
    log-supermodular, then we show that $\nCSP(\nfuns)$ is as easy as
    a Boolean log-supermodular weighted \nCSP{}.  Otherwise, we
    show that it is \NP{}-hard to approximate.  Finally, we give a
    full trichotomy for the arity-2 case, where $\nCSP(\nfuns)$ is in
    \FP{}, or is \nBIS{}-equivalent, or is equivalent in difficulty to
    $\nSAT$, the problem of approximately counting the satisfying
    assignments of a Boolean formula in conjunctive normal form.
    We also discuss the algorithmic aspects of our classification.
\end{abstract}

\section{Introduction}

In the weighted counting constraint satisfaction problem, there is a
fixed finite domain~$D$ and a fixed finite ``weighted constraint language''
$\nfuns$, which is a set of functions.  Every function $F\in \nfuns$
maps a tuple of domain elements to a value called a ``weight''. In the
computational problem $\nCSP(\nfuns)$, an instance consists of a set
$V=\{v_1,\ldots,v_n\}$ of variables and a set of ``weighted
constraints''. Each weighted constraint applies a function
from~$\nfuns$ to an appropriate-sized tuple of variables.

For example, with the Boolean domain $D=\{0,1\}$ we could consider the
situation in which $\nfuns$ consists of the single binary (arity-2) function~$F$
defined by $F(0,0)=F(0,1)=F(1,0)=1$ and $F(1,1)=2$.  We can construct
an instance with variables $v_1$, $v_2$ and $v_3$ and weighted
constraints $F(v_1,v_2)$ and $F(v_2,v_3)$.  If $\vecx$ is an
assignment of domain elements to the variables then the total weight
associated with~$\vecx$ is the product of the weighted constraints,
evaluated at~$\vecx$.
 
For example, the assignment that maps $v_1$, $v_2$ and $v_3$ all
to~$0$ has weight $F(0,0)F(0,0)=1$ but the assignment that maps all of
them to~$1$ has weight $F(1,1)F(1,1)=4$. Two assignments have weight
$F(1,1)F(1,0)=F(0,1)F(1,1)=2$ and the other four assignments each have
weight~$1$.  The computational problem is to evaluate the sum of the
weights of the assignments.  For this instance, the solution is~$13$.
 
There has been a lot of work on classifying the computational
difficulty of exactly solving $\nCSP(\nfuns)$. For some weighted
constraint languages
$\nfuns$, this is a computationally easy task, while for others, it is
intractable.  We will give a brief summary of what is known. For more
details, see the surveys of Chen~\cite{Chensurvey} and
Lu~\cite{Lusurvey}.

First, suppose that the domain $D$ is Boolean (that is, suppose that
$D=\{0,1\}$).  For this case, Creignou and Hermann~\cite{CH} gave a
dichotomy for the case in which weights are also in $\{0,1\}$. In this
case, they showed that $\nCSP(\nfuns)$ is in \FP{} (the set of
polynomial-time computable function problems) if all of the functions
in $\nfuns$ are affine, and that otherwise, it is
\numP{}-complete. Dyer, Goldberg, and Jerrum~\cite{wbool} extended
this to the case in which weights are non-negative rationals.  For
this case, they showed that the problem is solvable in polynomial time
if (1) every function in $\nfuns$ is expressible as a product of unary
functions, equalities and disequalities, or (2) every function
in $\nfuns$ is a constant multiple of an affine function.  Otherwise,
they showed the problem is complete in the complexity
class~$\FP^{\numP}\!$.  We will not deal with negative weights in this
paper. However, it is worth mentioning that these results have been
extended to the case in which weights can be
negative~\cite{negweightbool}, to the case in which they can be
complex~\cite{complexbool}, and to the related class of Holant$^*$
problems \cite{holantstar}. Other dichotomies are also known for
Holant problems (see \cite{Lusurvey}).

Next, consider an arbitrary finite domain~$D$.  For the case in which
weights are in $\{0,1\}$, Bulatov's breakthrough result
\cite{BulatovDichotomy} showed that $\nCSP(\nfuns)$ is always either
in \FP{} or \numP{}-hard.  A simplified version was given by Dyer and
Richerby 
\cite{DReffective},
who introduced a new criterion called
``strong balance''.  The dichotomy was extended to include
non-negative rational weights by Bulatov, Dyer, Goldberg,  Jalsenius, Jerrum and Richerby~\cite{BDGJJR} and then
to include all non-negative algebraic weights by Cai, Chen and Lu
\cite{NonNegExact, NonNegExactconf}.  Cai, Chen and Lu gave a
 generalised  notion of
balance that we will use in this work.  Finally, Cai and Chen
\cite{CCcomplex} extended the dichotomy to include all algebraic
complex weights.  The criterion for the unweighted \nCSP{} dichotomy
is known to be decidable~\cite{DReffective} and this carries through
to non-negative rational weights and non-negative algebraic
weights~\cite{NonNegExact}.  Decidability is currently open for the
complex case.

Much less is known about the complexity of approximately solving
$\nCSP(\nfuns)$.  Before describing what is known, it helps to say a
little about the complexity of approximate counting within \numP{}.
Dyer, Goldberg, Greenhill and Jerrum~\cite{DGGJ} identified three
complexity classes for approximation problems within \numP{}. These
are:
\begin{enumerate}
\item problems that have a fully polynomial randomised
    approximation scheme (FPRAS),
\item a logically defined complexity class called \RHPi{}, and
\item a class of problems for which approximation is \NP{}-hard. 
\end{enumerate} 
A typical complete problem in the class \RHPi{} is \nBIS{}, the
problem of approximately counting independent sets in bipartite
graphs. It is known that either all complete problems in \RHPi{} have an
FPRAS, or none do; it is conjectured that none  do~\cite{GJPotts}.  A
typical complete problem in the third class is $\nSAT$, the problem of
counting satisfying assignments of a Boolean formula in conjunctive
normal form.  Another concept that turns out to be important in the
classification of approximate counting CSPs is 
log-supermodularity~\cite{LSM}.  
A function with Boolean domain is
log-supermodular if its logarithm is supermodular; we give a formal
definition later.
 
Given those rough ideas, we can now describe what is known about the
complexity of approximately solving $\nCSP(\nfuns)$ when the domain,
$D$, is Boolean.  For the case in which weights are in $\{0,1\}$,
Dyer, Goldberg and Jerrum gave a trichotomy~\cite{DGJ}.  If every
function in $\nfuns$ is affine, then $\nCSP(\nfuns)$ is in \FP{}.
Otherwise, it is as hard to approximate as \nBIS{}.  The hard
approximation problems are divided into \nBIS{}-equivalent cases
(which arise when the functions in~$\nfuns$ can be expressed using
``implies'' and fixing the values of certain variables) and the
remaining cases, which are all shown to be $\nSAT$-equivalent.

In the more general case where the domain~$D$ is still Boolean, but
the weights can be arbitrary non-negative values, no complete
classification is known.  However, Bulatov, Dyer, Goldberg, Jerrum and
McQuillan~\cite{LSM}
gave a classification for the
so-called ``conservative'' case, in which $\nfuns$ contains all unary
functions.  Their result is reproduced as Lemma~\ref{lem:lsm_mainthm}
below.  Here is an informal description.
\begin{itemize}
\item If every function in $\nfuns$ can be expressed in a certain
    simple way using disequality and unary functions, then, for any
    finite $\nfinfuns\subset\nfuns$, $\nCSP(\nfinfuns)$ has an FPRAS.
\item Otherwise,
    \begin{itemize}
    \item there is a finite $\nfinfuns\subset\nfuns$ such that
        $\nCSP(\nfinfuns)$ is at least as hard to approximate as
        \nBIS{} and,
    \item if $\nfuns$ contains any function that is not
        log-supermodular, then there is a finite $\nfinfuns\subset
        \nfuns$ such that $\nCSP(\nfinfuns)$ is at least as hard to
        approximate as \nSAT{}.
    \end{itemize}
\end{itemize}
Yamakami~\cite{Tomo} has also given an
approximation dichotomy for the case in which even more unary
functions (including those with negative weights) are assumed to be
part of $\nfuns$. The negative weights introduce cancellation, making
more weighted constraint languages $\nfuns$ intractable. In this
paper, we stick to the
non-negative case, in which more subtle complexity classifications
arise.
 
Prior to this paper, there were no known complexity classifications
for approximately solving $\nCSP(\nfuns)$ for the case in which the
domain~$D$ is not Boolean.  Thus, this is the problem that we address
in this paper.  Our main result (Theorem~\ref{thm:main}, below) is a
complexity classification for the conservative case (where all unary
weights are contained in~$\nfuns$).  Here is an informal description
of the result.
\begin{itemize}
\item If $\nfuns$ is ``weakly log-modular'' (a concept we define
    below) then, for any finite $\nfinfuns \subset \nfuns$,
    $\nCSP(\nfinfuns)$ is in \FP{}.
\item Otherwise, there is a finite $\nfinfuns \subset \nfuns$ such that
    $\nCSP(\nfinfuns)$ is at least as hard to approximate as \nBIS{}.
    Furthermore,
    \begin{itemize}
    \item if $\nfuns$ is ``weakly log-supermodular'' (again, defined
        below) then, for any finite $\nfinfuns \subset \nfuns$, there
        is a finite set $\nfinfuns'$ of log-supermodular functions on
        the Boolean domain such that $\nCSP(\nfinfuns)$ is as easy to
        approximate as $\nCSP(\nfinfuns')$;
    \item otherwise, there is a finite $\nfinfuns \subset \nfuns$ such
        that $\nCSP(\nfinfuns)$ is as hard to approximate as \nSAT{}.
    \end{itemize}
\end{itemize}

Informally, $\nfuns$ is weakly log-supermodular if, for every binary
function~$F$ that can be expressed using functions in~$\nfuns$, every
projection of~$F$ onto two domain elements is log-supermodular (see
Definition~\ref{def:weaklylogsupermod}).  Thus, in some sense, our
result shows that all the difficulty of approximating conservative
weighted constraint satisfaction
problems arises in the Boolean case.  Even when the domain~$D$ is
larger, approximations which are $\nSAT$-equivalent are
$\nSAT$-equivalent precisely because of intractable Boolean problems
which arise as sub-problems.

In addition to the complexity classifications described above (\FP{}
versus \nBIS{}-hard and ``as easy as a Boolean log-supermodular
problem'' versus $\nSAT$-equivalent) we also give a full trichotomy
for the binary case (i.e., where all functions in $\nfuns$ have arity
1 or~2).
\begin{itemize}
\item If $\nfuns$ is weakly log-modular then, for any finite
    $\nfinfuns \subset \nfuns$, $\nCSP(\nfinfuns)$ is in \FP{}.
\item Otherwise, if $\nfuns$ is weakly log-supermodular, then
    \begin{itemize}
    \item for every finite $\nfinfuns \subset \nfuns$,
        $\nCSP(\nfinfuns)$ is as easy to approximate as \nBIS{} and
    \item there is a finite $\nfinfuns \subset \nfuns$ such that
        $\nCSP(\nfinfuns)$ is as hard to approximate as \nBIS{}.
    \end{itemize}
\item Otherwise, there is a finite $\nfinfuns \subset \nfuns$ such
    that $\nCSP(\nfinfuns)$ is as hard to approximate as \nSAT{}.
\end{itemize}

The final section of our paper discusses the algorithmic aspects of
our classification for the case in which $\nfuns$ is the union of a
finite, weighted constraint language $\nhfuns$ and the set of all
unary functions.  In particular, we give an algorithm that takes
$\nhfuns$ as input and correctly makes one of the following
deductions:
\begin{enumerate}
\item $\nCSP(\nfinfuns)$ is in \FP{} for every finite
    $\nfinfuns\subset \nfuns$;
\item $\nCSP(\nfinfuns)$ is \LSM{}-easy for every finite $\nfinfuns
    \subset \nfuns$ and \nBIS{}-hard for some such~$\nfinfuns$;\label{deduc}
\item $\nCSP(\nfinfuns)$ is \nBIS{}-easy for every finite
    $\nfinfuns \subset \nfuns$ and \nBIS{}-equivalent for some
    such~$\nfinfuns$;
\item $\nCSP(\nfinfuns)$ is \nSAT{}-easy for all finite $\nfinfuns
    \subset \nfuns$ and \nSAT{}-equivalent for some such~$\nfinfuns$.
\end{enumerate}
Further, if every function in $\nhfuns$ has arity at most~2, the
output is not deduction~\ref{deduc}.
The term ``\LSM{}-easy'' in deduction~\ref{deduc} will be formally defined later.
Informally, it means ``as easy as a Boolean log-supermodular weighted counting CSP''.

\subsection{Previous work} 
 
The first contribution of our paper is to show that, if $\nfuns$ is
weakly log-modular then, for any finite $\nfinfuns \subset \nfuns$,
$\nCSP(\nfinfuns)$ is in \FP{}.  Otherwise, there is a finite
$\nfinfuns \subset \nfuns$ for which $\nCSP(\nfinfuns)$ is at least as
hard to approximate as \nBIS{}.  We also show that, if $\nfuns$ is not
weakly log-supermodular, then there is a finite $\nfinfuns \subset
\nfuns$, such that $\nCSP(\nfinfuns)$ is \nSAT{}-hard.  This work is
presented in Sections \ref{sec:hardness} and~\ref{sec:balance} below
and builds on two strands of previous work.
\begin{itemize}
\item The hardness results build on the approximation
    classification in the Boolean case~\cite{LSM} and, in particular,
    on the key role played by log-supermodular functions.
\item The easiness results build on the classification of the exact
    evaluation of $\nCSP(\nfuns)$ in the general
    case~\cite{NonNegExact}, and in particular on the key role played
    by ``balance''.
\end{itemize}

The second contribution of our paper is to show that, if $\nfuns$ is
weakly log-supermodular, then, for any finite $\nfinfuns \subset
\nfuns$, there is a finite set $\nfinfuns'$ of log-supermodular
functions on the Boolean domain such that $\nCSP(\nfinfuns)$ is as
easy to approximate as $\nCSP(\nfinfuns')$.  This builds on
three key studies of the complexity of optimisation CSPs by
Takhanov~\cite{TakhanovFull, Takhanovconf}, Cohen, Cooper and Jeavons~\cite{CCJ} and
Kolmogorov and \Zivny~\cite{KZ}.  In all three cases, we use
their arguments and ideas, and not merely their results. Thus, we
delve into these three papers in some detail.

Our final contribution is the trichotomy for the binary case.  This
relies additionally on work of  
Cohen, Cooper, Jeavons and Krokhin~\cite{new} generalising
Rudolf and Woeginger's decomposition~\cite{RudolfWoeginger} of  Monge matrices.
The Monge property can be viewed
as a  generalisation of binary submodular functions to a larger domain, and
the decomposition shows how to decompose such functions 
in a useful manner.  

\subsection{Preliminaries and statement of results}\label{sec:prelims}

Let $D$ be a finite domain with $|D| \geq 2$.  It will be convenient
to refer to the set $\funsk R$ of all functions $D^k\to R$ for some
codomain $R$, and the set $\funs R=\bigcup_{k=0}^\infty \funsk R$.
Let $\EQ$ be the binary equality function defined by $\EQ(x,x)=1$ and
$\EQ(x,y)=0$ for $x\neq y$; let $\NEQ(x,y) = 1 - \EQ(x,y)$.

We use the following definitions from~\cite{LSM}.  Let $\nfuns$ be a
subset of $\funs R$.  Let $V=\{v_1,\dots,v_n\}$ be a set of variables.
An atomic formula has the form $\phi=G(v_{i_1},\dots,v_{i_a})$ where
$G\in\nfuns$, $a=a(G)$ is the arity of~$G$, and $(v_{i_1}, v_{i_2},
\dots, v_{i_a})\in V^a$ is called a ``scope''.  Note that repeated
variables are allowed.  The function $F_\phi\colon D^n \to R$
represented by the atomic formula $\phi=G(v_{i_1},\dots,v_{i_a})$ is
just $F_\phi(\vecx)=G(\vecx(v_{i_1}),\dots,\vecx(v_{i_a}))$, where
$\vecx\colon\{v_1,\dots,v_n\}\to\dom$ is an assignment to the
variables.  To simplify the notation, we write $x_j = \vecx(v_j)$ so
\begin{equation*}
    F_\phi(\vecx)= G(x_{i_1},\dots,x_{i_a}).
\end{equation*}

A pps-formula (``primitive product summation formula'') is a finite summation
of a finite product of atomic formulas.  A pps-formula~$\psi$ over~$\nfuns$
in variables $V'=\{v_1,\dots,v_{n+k}\}$ has the form
\begin{equation*}
    \psi=\sum_{v_{n+1},\dots,v_{n+k}}\,\prod_{j=1}^m\phi_j,
\end{equation*}
where $\phi_j$ are all atomic formulas over~$\nfuns$ in the
variables~$V'\!$.  (The variables $V$ are free, and the others,
$V'\setminus V$, are bound.)  The formula~$\psi$ specifies a function
$F_\psi\colon D^n\to R$ in the following way:
\begin{equation*} 
    F_\psi(\vecx)=\sum_{\vecy\in\dom^k}\prod_{j=1}^m F_{\phi_j}(\vecx,\vecy),
\end{equation*}
where $\vecx$ and $\vecy$ are assignments $\vecx\colon
\{v_1,\dots,v_n\} \to \dom$ and $\vecy\colon \{v_{n+1}, \dots,
\allowbreak v_{n+k}\} \to \dom$.  The \emph{functional
  clone~$\fclone\nfuns$ generated by~$\nfuns$} is the set of all
functions that can be represented by a pps-formula over
$\nfuns\cup\{\EQ\}$.  Crucially,
$\fclone{\fclone{\nfuns}}=\fclone{\nfuns}$ 
\cite[Lemma~2.1]{LSM}; 
we
will rely on this transitivity property implicitly.

\begin{definition}  
    A \emph{weighted constraint language} $\nfuns$ is a subset of
    $\funs{\posQ}$.  Functions in $\nfuns$ are called \emph{weight
      functions}.
\end{definition}

In Section \ref{sec:vcsp} we will introduce \emph{valued constraint
  languages} and \emph{cost functions}, which pertain to optimisation
CSPs.  It is important to distinguish these from the weighted version,
which is used for counting.
  
\begin{definition}
    A weighted constraint language $\nfuns$ is \emph{conservative} if
    $\uns{D}\subseteq\nfuns$, where $\uns{D}=\func_1(D,\posQ)$.
\end{definition} 

\begin{definition} 
    A weighted constraint language $\nfuns$ is \emph{weakly
      log-modular} if, for all binary functions $F\in\fclone{\nfuns}$
    and elements $a,b\in D$,
    \begin{align}
    \label{eq:weaklm}
        &F(a,a)F(b,b)=F(a,b)F(b,a),\text{ or }\nonumber\\
        &F(a,a)=F(b,b)=0,\text{ or }\nonumber\\
        &F(a,b)=F(b,a)=0.
    \end{align}
\end{definition}

\begin{definition}\label{def:weaklylogsupermod}
    $\nfuns$ is \emph{weakly log-supermodular} if, for all binary
    functions $F\in\fclone\nfuns$ and elements $a,b\in D$,
    \begin{equation}
    \label{eq:weaklsm}
        F(a,a)F(b,b)\geq F(a,b)F(b,a)\quad\text{or}\quad F(a,a)=F(b,b)=0.
    \end{equation}
\end{definition}

\begin{definition} 
\label{def:fnLSM}
    A function $F\in\func_k(\{0,1\},\posQ)$ is \emph{log-supermodular}
    if
    \begin{equation*}
        F(\vecxuy)F(\vecxny)\geq F(\vecx)F(\vecy)
    \end{equation*}
    for all $\vecx,\vecy\in\{0,1\}^k\!$, where $\wedge$ (min) and
    $\vee$ (max) are
    applied component-wise.  \LSM{} is the set of all log-supermodular
    functions in $\func(\{0,1\},\posQ)$.
\end{definition}

It is known  
\cite[Lemma~4.2]{LSM} 
that $\fclone{\LSM}=\LSM$.  Here is a
precise statement of the computational task that we study.  A \nCSP{}
problem is parameterised by a finite, weighted constraint
language $\nfuns$ as follows.

\prob{$\nCSP(\nfuns)$.}{A pps-formula $\psi$ consisting of a product of
  $m$~atomic $\nfuns$-formulas over $n$~free variables~$\vecx$.
  (Thus, $\psi$ has no bound variables.)}{The value $\sum_{\vecx\in
    D^n} F_\psi(\vecx)$ where $F_\psi$ is the function defined by
  $\psi$.}

Where convenient, we abuse notation by writing $\nCSP(F)$ to mean
$\nCSP(\{F\})$ and by writing $\nCSP(\nfuns,\nfuns')$ to mean
$\nCSP(\nfuns\cup \nfuns')$.
 
As in \cite{LSM} (and other works) we take the size of a
$\nCSP(\nfuns)$ instance to be $n+m$, where $n$ is the number of
(free) variables and $m$ is the number of weighted constraints (atomic
formulas).  In unweighted versions of CSP and \nCSP{}, we can just use
$n$ as the size of an instance, since the number of constraints can be
bounded by a polynomial in the number of variables.  However, in
weighted cases, the multiplicity of constraints matters so we cannot
bound $m$ in terms of $n$.  We typically denote an instance of
$\nCSP(\nfuns)$ by~$I$ and the output by $Z(I)$, which is often called
the ``partition function''.
 
A counting problem, for our purposes, is any function from instances
(encoded as words over a finite alphabet $\Sigma$) to $\posQ$.  A
\emph{randomised approximation scheme\/} for a counting problem~$\#X$
is a randomised algorithm that takes an instance~$w$ and returns an
approximation $Y$ to $\#X(w)$.  The approximation scheme has a
parameter~$\epsilon\in(0,1)$ which specifies the error tolerance.  Since the
algorithm is randomised, the output~$Y$ is a random variable depending
on the ``coin tosses'' made by the algorithm.  We require that, for
every instance~$w$ and every $\epsilon\in(0,1)$,
\begin{equation}
\label{eq:FPRASerrorprob}
    \Pr \big[e^{-\epsilon} \#X(w)\leq Y \leq e^\epsilon \#X(w)\big]\geq  3/4\,.
\end{equation}
The randomised approximation scheme is said to be a \emph{fully
  polynomial randomised approximation scheme}, or \emph{FPRAS}, if it
runs in time bounded by a polynomial in~$|w|$ (the length of the
word~$w$) and $\epsilon^{-1}\!$.  See Mitzenmacher and
Upfal~\cite[Definition 10.2]{MU05}.

Suppose that $\#X$ and $\#Y$ are two counting problems. An
``approx\-imation-preserving reduction'' (AP-reduction)~\cite{DGGJ}
from~$\#X$ to~$\#Y$ gives a way to turn an FPRAS for~$\#Y$ into an
FPRAS for~$\#X$.  Specifically, an {\it AP-reduction from $\#X$
  to~$\#Y$} is a randomised algorithm~$\mathcal{A}$ for
computing~$\#X$ using an oracle for~$\#Y$.  The
algorithm~$\mathcal{A}$ takes as input a pair
$(w,\epsilon)\in\Sigma^*\times(0,1)$, and satisfies the following
three conditions: (i)~every oracle call made by~$\mathcal{A}$ is of
the form $(v,\delta)$, where $v\in\Sigma^*$ is an instance of~$\#Y$,
and $0<\delta<1$ is an error bound satisfying $\delta^{-1} \leq
\mathrm{poly}(|w|, \epsilon^{-1})$; (ii)~the algorithm~$\mathcal{A}$
meets the specification for being a randomised approximation scheme
for~$\#X$ (as described above) whenever the oracle meets the
specification for being a randomised approximation scheme for~$\#Y$;
and (iii)~the run-time of~$\mathcal{A}$ is polynomial in $|w|$ and
$\epsilon^{-1}\!$.  If an AP-reduction from $\#X$ to~$\#Y$ exists we
write $\#X\APred\#Y$.  Note that, subsequent to~\cite{DGGJ}, the
notation~$\APred$ has been used to denote a different type of
approximation-preserving reduction which applies to optimisation
problems. In this paper, our emphasis is on counting problems so
we hope this will not cause confusion.
 
The notion of pps-definability described earlier is closely related to
AP-reductions.  In particular,  
\cite[Lemma~10.1]{LSM} 
shows that
$G\in\fclone{\nfuns}$ implies that $\nCSP(\nfuns,G)\APred\nCSP(\nfuns)$.
We will use this fact without comment.

As mentioned above, \nBIS{} is the problem of counting
the independent sets in a bipartite graph and $\nSAT$ is the
problem of counting the solutions to a Boolean formula in
conjunctive normal form.  We say that a counting problem $\#X$ is
$\#Y$-easy if $\#X\APred\#Y$ and that it is $\#Y$-hard if
$\#Y\APred\#X$.  A problem $\#X$ is \emph{\LSM{}-easy} if there is a
finite, weighted constraint language $\nfuns\subset\LSM$ such that
$\#X\APred\nCSP(\nfuns)$.

We now state our main theorem.  Note that we have only defined the
problem $\nCSP(\nfuns)$ for finite languages whereas conservative
languages are, by definition, infinite.

\begin{theorem}\label{thm:main}
    Let $\nfuns$ be a conservative weighted constraint language
    taking values in $\posQ$.
    \begin{enumerate}
    \item If $\nfuns$ is weakly log-modular then $\nCSP(\nfinfuns)$ is
        in \FP{} for every finite $\nfinfuns\subset\nfuns$.
    \item If $\nfuns$ is weakly log-supermodular but not weakly
        log-modular, then $\nCSP(\nfinfuns)$ is \LSM{}-easy for every
        finite $\nfinfuns\subset\nfuns$ and \nBIS{}-hard for some such
        $\nfinfuns$.
    \item If $\nfuns$ is weakly log-supermodular but not weakly
        log-modular and consists of functions of arity at most two,
        then $\nCSP(\nfinfuns)$ is \nBIS{}-easy for every finite
        $\nfinfuns\subset\nfuns$ and \nBIS{}-equivalent for some such
        $\nfinfuns$.
    \item If $\nfuns$ is not weakly log-supermodular, then
        $\nCSP(\nfinfuns)$ is \nSAT{}-easy for every finite
        $\nfinfuns\subset\nfuns$ and \nSAT{}-equivalent for some such
        $\nfinfuns$.
    \end{enumerate}
\end{theorem}

In particular, among conservative \nCSP{}s, there are no new
complexity classes below \nBIS{} or above \LSM{}; furthermore there is
a trichotomy for conservative weighted constraint languages with no
functions of arity greater than two.

The \nBIS{}-hardness and \nSAT{}-equivalence are proved in
Section~\ref{sec:hardness}, where they are restated as
Theorem~\ref{thm:BIShard-SATequiv}.  The membership in \FP{} is
established as Theorem~\ref{thm:tractable} at the end of
Section~\ref{sec:balance}. \LSM{}-easiness and \nBIS{}-easiness are
established by Theorem~\ref{thm:LSM-BIS-easy} at the end of
Section~\ref{sec:lsm_bis_easy}.  Algorithmic aspects are discussed in
Section~\ref{sec:algorithmic}.

\section{Hardness results}
\label{sec:hardness}

Our hardness results use the following result from~\cite{LSM}.

\begin{lemma} 
\cite[Theorem 10.2]{LSM}
\label{lem:lsm_mainthm}
    Let $\nfuns$ be a finite, weighted constraint language with
    $D=\{0,1\}$. 
    \begin{itemize}
    \item If $\nfuns\subset\fclone{\NEQ,\uns{\{0,1\}}}$ then, for
        any finite $S\subset \uns{\{0,1\}}$, $\nCSP(\nfuns\cup S)$ has
        an FPRAS.
    \item If $\nfuns\not\subset\fclone{\NEQ,\uns{\{0,1\}}}$, then
        there is a finite $S\subset \uns{\{0,1\}}$ such that
        $\nCSP(\nfuns\cup S)$ is \nBIS{}-hard.
    \item If $\nfuns\not\subset\fclone{\NEQ,\uns{\{0,1\}}}$ and
        $\nfuns\not\subset\LSM$, then there is a finite $S\subset
        \uns{\{0,1\}}$ such that $\nCSP(\nfuns\cup S)$ is
        \nSAT{}-hard.
    \end{itemize}
\end{lemma}

\begin{remark}\label{rmk:rationallsm}
    In the statement of  
    \cite[Theorem 10.2]{LSM}, the set
    $\uns{\{0,1\}}$ is replaced with $\mathcal{B}^{p}_1$, the set of
    all unary functions from $\{0,1\}$ to the set of non-negative
    efficiently computable reals. In this paper, we restrict to
    rationals for simplicity. Even though the statement of 
    \cite[Theorem 10.2]{LSM} does not imply Lemma~\ref{lem:lsm_mainthm},     the proof of \cite[Theorem 10.2]{LSM} does establish the lemma.  No
    functions in $\mathcal{B}^{p}_1$ with irrational weights are used
    explicitly in the proof --- unary functions that are used (for
    example, in the proof of  
    \cite[Lemma 7.1]{LSM}) are constructed by
    multiplying and dividing other values in the codomains of
    functions in $\nfuns$.
\end{remark}

In fact we will only use the following special case of
Lemma~\ref{lem:lsm_mainthm}.

\begin{lemma} 
\cite[Theorem 10.2]{LSM}
\label{lem:lsm_mainthm_special}
    Let $F$ be a function in $\func_2(\{0,1\},\posQ)$.
    \begin{itemize}
    \item If $F\notin\fclone{\NEQ,\uns{\{0,1\}}}$ then
        $\{F\}\cup\uns{\{0,1\}}$ is \nBIS{}-hard.
    \item If $F\notin\fclone{\NEQ,\uns{\{0,1\}}}\cup\LSM$ then
        $\{F\}\cup\uns{\{0,1\}}$ is \nSAT{}-hard.
    \end{itemize}
\end{lemma}

Our hardness results now follow from Lemma~\ref{lem:lsm_mainthm_special}.

\begin{theorem}
\label{thm:BIShard-SATequiv}
    Let $\nfuns$ be a conservative weighted constraint language taking
    values in $\posQ$.
    \begin{itemize}
    \item If $\nfuns$ is not weakly log-modular, there is a finite
        $\nfinfuns\subset\nfuns$ such that $\nCSP(\nfinfuns)$ is
        \nBIS{}-hard.
    \item If $\nfuns$ is not weakly log-supermodular, there is a
        finite $\nfinfuns\subset\nfuns$ such that $\nCSP(\nfinfuns)$
        is \nSAT{}-hard.
    \item For all finite $\nfinfuns\subset\nfuns$, $\nCSP(\nfinfuns)$
        is \nSAT{}-easy.
    \end{itemize}
\end{theorem}
\begin{proof}
    First, we establish the hardness results.
 
    Suppose that $\nfuns$ is not weakly log-modular.  Let
    $H\in\fclone{\nfuns}$ be a function violating \eqref{eq:weaklm}
    and let $a$ and $b$ be the relevant domain elements, which must be
    distinct.
    Let $\phi\colon\{0,1\}\to D$ be a unary function with
    $\phi(0)=a$ and $\phi(1)=b$.  Define $H_\phi\colon\{0,1\}^2\to\posQ$ by
    $H_\phi(x,y)=H(\phi(x),\phi(y))$.  The following three equations
    must all fail to hold:
    \begin{gather*} 
        H_\phi(0,0)H_\phi(1,1)=H_\phi(0,1)H_\phi(1,0)\\
        H_\phi(0,0)=H_\phi(1,1)=0 \\
        H_\phi(0,1)=H_\phi(1,0)=0.
    \end{gather*}
    By      \cite[Remark 7.3]{LSM}, every binary function in
    $\fclone{\NEQ,\uns{\{0,1\}}}$ has one of three forms: $U_1(x)
    U_2(y)$, $U(x) \EQ(x,y)$ or $U(x) \NEQ(x,y)$. Therefore, $H_\phi
    \notin\fclone{\NEQ,\uns{\{0,1\}}}$.  By Lemma
    \ref{lem:lsm_mainthm_special} there is a finite set
    $S\subset\uns{\{0,1\}}$ such that $\nBIS\APred\nCSP(H_\phi,S)$.

    For each $U \in \uns{\{0,1\}}$, define $U_{\phi^{-1}}\in\uns D$ by
    \begin{align*}
        U_{\phi^{-1}}(x)=\begin{cases}
            U(0) & \text{ if $x=a$,}\\
            U(1) & \text{ if $x=b$,}\\
            0    & \text{ otherwise.} 
        \end{cases} 
    \end{align*}
    Let $E(0)=E(1)=1$.  Let $S'=\{U_{\phi^{-1}}\mid U\in S\cup
    \{E\}\}$.  Note that $\{H\}\cup S'\subset \fclone{F,\uns{D}}$ is
    finite.

    We describe a reduction from $\nCSP(H_\phi,S)$ to
    $\nCSP(H,S')$.  Given an instance $I$ of $\nCSP(H_\phi,S)$, replace
    each use of $H_\phi$ by $H$, and each use of $U\in S$ by
    $U_{\phi^{-1}}\in S'\!$, and introduce an atomic formula
    $E_{\phi^{-1}}(v)$ for each variable $v$, to obtain a new instance
    $I'$ of $\nCSP(H,S')$ with $Z(I)=Z(I')$.  Thus $\nCSP(H, S')$ is
    \nBIS{}-hard.

    A similar argument shows that $\nfuns$ is \nSAT{}-hard if
    it is not weakly log-supermodular. In this
    case, we start with a function $H\in \fclone{\nfuns}$ violating
    \eqref{eq:weaklsm} on the elements $a,b\in \dom$.  Defining $\phi$
    and $H_\phi$ as above, we find that $H_\phi\not\in\LSM$.  Since
    $H$ also violates \eqref{eq:weaklm} on $a,b$, the argument above
    establishes $H_\phi \notin\fclone{\NEQ,\uns{\{0,1\}}}$.  By Lemma
    \ref{lem:lsm_mainthm_special} there is a finite set
    $S\subset\uns{\{0,1\}}$ such that $\nSAT\APred\nCSP(H_\phi,S)$.
    We then proceed as before.
 
    \nSAT{}-easiness follows from the construction in Section~3
    of~\cite{DGGJ}, which shows that any problem in \numP{} is
    \nSAT{}-easy.  The weighted counting CSPs that we deal with here
    are equivalent, by \cite{BDGJJR}, to unweighted ones, which are in
    \numP{}.
\end{proof}

\section{Balance and weak log-modularity}\label{sec:balance}

In this section we show that weak log-modularity implies tractability,
by showing that every weakly log-modular weighted constraint language
is balanced in the following sense.

We may associate a matrix $\MM$ with an undirected bipartite graph
$G_\MM$ whose vertex partition consists of the set of rows $R$ and
columns $C$ of $\MM$.  A pair $(r,c)\in R\times C$ is an edge of
$G_\MM$ if, and only if, $\MM_{rc}\neq 0$.  A \emph{block} of $\MM$ is
a submatrix whose rows and columns form a connected component in
$G_\MM$.  $\MM$ has block-rank~$1$ if all its blocks have rank~$1$.
 
We say that a weighted constraint language $\nfuns$ is \emph{balanced}
\cite{NonNegExact} if, for every function $F(x_1, \dots, x_n)\in
\fclone{\nfuns}$ with arity $n\geq 2$, and every $k$ with $0<k<n$, the
$|\dom|^k\times |\dom|^{n-k}$ matrix $F((x_1,\dots, x_k), (x_{k+1},
\dots, x_n))$ has block-rank~$1$.
(This notion reduces to Dyer and Richerby's notion of ``strong
balance''\cite{DReffective} in the unweighted case.)
 
A function $F\colon\{0,1\}^n\to\mathbb{R}$ is \emph{strictly positive} if
its range is contained in $\mathbb{R}_{>0}$.
$F$
has \emph{rank~$1$} if it has
the form $F(x_1,\dots,x_k)=U_1(x_1)\cdots U_k(x_k)$.
Given a non-singular two-by-two matrix
 $T\in\mathbb{R}^{2\times 2}$  
  we define
$T^{\otimes n}F\colon\{0,1\}^n\to\mathbb{R}$ by
\begin{equation*}
    (T^{\otimes n}F)(x_1,\dots,x_n)
        = \sum_{y_1,\dots,y_n\in\{0,1\}}
              \left(\prod_{i=1}^n T_{x_iy_i}\right) F(y_1,\dots,y_n)
\end{equation*}
The rows and columns of $T$ are considered to be indexed by
$\{0,1\}$. We associate with any function $F\colon\{0,1\}^2\to\mathbb{R}$, the
matrix $M_F\in\mathbb{R}^{2\times 2}$ defined by $(M_F)_{ij}=F(i,j)$.

\begin{lemma}\label{lem:bool_facts}
    Let $\MM\in\mathbb{R}^{k\times k}$. Let $F\colon \{0,1\}^k \to
    \mathbb{R}$. Let $T\in\mathbb{R}_{\geq 0}^{2\times 2}$ be
    non-singular.
    \begin{enumerate}

    \item \label{pt:binary}  
 If $k=2$ then
        $\MM$ has
        block-rank~$1$ if and only if it has rank~$1$ or it has at most two
        non-zero entries.   
$F$~has
rank~$1$
        if and only if $\det M_F=0$.

    \item \label{pt:br1_2by2} $\MM$ has block-rank~$1$ if and only if
        the matrix
        \begin{equation*}
            N_{M,\vecu,\vecu'\!,\vecv,\vecv'}
                = \begin{pmatrix}
                      \MM(\vecu,\vecv)&\MM(\vecu,\vecv')\\
                      \MM(\vecu'\!,\vecv)&\MM(\vecu'\!,\vecv')
                  \end{pmatrix}
        \end{equation*}
        has block-rank~$1$ for every $\vecu, \vecu'\!, \vecv, \vecv'\!$.

    \item \label{pt:topkis} (Topkis's theorem) If $F$ is strictly
        positive and not of rank~$1$, there is a function $F'\colon \{0,1\}^2
        \to \mathbb{R}$ of the following form that is not of rank~$1$.
        \begin{equation*}
            F'(x_i,x_j) =
                F(c_1,\dots,c_{i-1},x_i,c_{i+1},\dots,c_{j-1},x_j,
                  c_{j+1},\dots,c_k).
        \end{equation*}
        Here $1\leq i< j\leq k$, and each $c_\ell$ is a fixed element of
        $\{0,1\}$.

    \item \label{pt:rk_basisfree} $F$ has rank~$1$ if and only if
        $T^{\otimes k}F$ has rank~$1$.
    \end{enumerate}
\end{lemma}
\begin{proof}

\begin{enumerate}
\item A $2\times 2$ matrix that has block-rank~$1$ either has rank~$1$ or
    is diagonal or anti-diagonal so has two zeroes.  Conversely, a
    matrix that has rank~$1$ has no submatrix whose rank exceeds~1, so
    has block-rank~$1$. A matrix with two or more zeroes has no $2\times
    2$ block so can only have blocks of rank~$1$.

    For the second statement, if $F$ has rank~$1$ then
    there are unary functions $U_0$ and $U_1$ so that
    $F(x,y) = U_0(x) U_1(y)$, which implies that $\det M_F= 0$.
    Going the other way, if $F$ is identically~$0$ then it has rank~$1$.
    Otherwise,  suppose $F(i,j)\neq 0$.
    Let $U_0(x)=F(x,j)$ and $U_1(y)=F(i,y)/F(i,j)$.
    If $\det M_F=0$ then $F(x,y)=U_0(x)U_1(y)$, so $F$ has rank~$1$.   
    
\item \cite[Lemma 38]{DReffective}.

\item Say that a strictly positive function
    $F$ is \emph{log-modular} if $f=\log F$ is modular: that is,
    $F(\vecx\vee\vecy)F(\vecx\wedge\vecy) = F(\vecx)F(\vecy)$ for all
    $\vecx,\vecy\in\{0,1\}^k\!$.  
If $f(\vecx)$ is modular, then
    can be expressed as a linear sum of the $x_i$'s
    (see for example \cite[Proposition 24]{BorosHammer}), so a
    strictly positive log-modular function is a product of unary
    functions, so it has rank~$1$. The result is then Topkis's
    theorem~\cite{Topkis} in the form stated in  
    \cite[Lemma 5.1]{LSM}.

\item If $F$ is of the form $U_1(x_1)\cdots U_n(x_n)$ then
    \begin{equation*}
        (T^{\otimes n}F)(x_1,\dots,x_n)
            = (T^{\otimes 1}U_1)(x_1)\cdots (T^{\otimes 1}U_n)(x_n)
    \end{equation*}
    The reverse implication follows from $(T^{-1})^{\otimes n}
    T^{\otimes n}F=F$, where $T^{-1}$ is the matrix inverse of~$T$.
\end{enumerate}
\end{proof}

A function $F\colon D^n\to\posQ$ is \emph{essentially pseudo-Boolean} if
its support 
(the set of vectors $\vecx$ satisfying $F(\vecx)>0$)
is contained in a set $\dom_1\times\dots\times \dom_n$
with $|\dom_1|,\dots,|\dom_n|\leq 2$.  The \emph{projection} of a
relation $R\subseteq D^n$ onto indices $1\leq i<j\leq n$ is the set of
pairs $(a,b)\in D^2$ such that there exists $\vecx\in R$ with $x_i=a$
and $x_j=b$.  A \emph{generalised $\NEQ$} is a relation of the form
$\{(x_i,x_j),(y_i,y_j)\}\subset D^2$ for some $x_i\neq y_i$ and
$x_j\neq y_j$.

\begin{lemma}
\label{lem:balance_wmod_helper}
    Let $F\colon\dom^n\to\posQ$ be an essentially pseudo-Boolean function
    which is not of rank~$1$, and assume that no binary projection of the
    support of $F$ is a generalised $\NEQ$. Then $\{F\}\cup \uns D$ is
    not weakly log-modular.
\end{lemma}
\begin{proof}
    Let the support of $F$ be contained in
    $\dom_1\times\dots\times\dom_n$ where $|\dom_i|=2$ for all $i$.
    Choose bijections $\rho_i\colon\{0,1\}\to \dom_i$ for each $1\leq
    i\leq n$.  Define $F_\rho\colon\{0,1\}^n\to\posQ$ by
    \begin{equation*}
        F_\rho(x_1,\dots,x_n)=F(\rho_1(x_1),\dots,\rho_n(x_n))
    \end{equation*}
    for all $x_1,\dots,x_n\in\{0,1\}$.  Let $T =
    \left(\begin{smallmatrix}2 & 1 \\ 1 & 2\end{smallmatrix}\right)$
    and note that $T^{\otimes n}F_\rho$ is strictly positive.  Since
    $F$ is not of rank~$1$, $F_\rho$ is not of rank~$1$, so by
    Lemma~\ref{lem:bool_facts} part~(\ref{pt:rk_basisfree}),
    $T^{\otimes n}F_\rho$ is not of rank~$1$.  By Lemma
    \ref{lem:bool_facts} part~(\ref{pt:topkis}), there is a function
    $B\colon\{0,1\}^2\to\posQ$ of the following form that is not of
    rank~$1$.
    \begin{equation*}
        B(x_i,x_j) = (T^{\otimes n} F_\rho)
        (c_1,\dots,c_{i-1},x_i,c_{i+1},\dots,c_{j-1},x_j,c_{j+1},\dots,c_n).
    \end{equation*}

    For all indices $k\in\{1,\dots,n\}\setminus\{i,j\}$, define
    $U_k\in\uns{D}$ by $U_k(\rho_k(x_k))=T_{c_k x_k}$ for all
    $x_k\in\{0,1\}$, and $U_k(z)=0$ if $z\notin D_k$.  Define
    $G,H\colon\dom^2\to\posQ$ and $G_{\rho_i,\rho_j},
    H_{\rho_i,\rho_i}\colon \{0,1\}^2\to\posQ$ as follows. Note in
    these definitions that~$i$ and~$j$ are fixed, but $\rho_i$ and
    $\rho_j$ are used as subscripts in the name of some of the
    functions as a reminder of the bijections that are being applied
    to the inputs.  Thus, in $H_{\rho_i, \rho_i}$, the bijection
    $\rho_i$ is applied to both arguments, even though the function
    depends on both $i$ and~$j$.
    \begin{align*}
        G(y_i,y_j)  &= \sum \left(
                           \prod_{k\neq i,j} U_k(y_k)
                       \right)F(y_1,\dots,y_n)
                    && \text{ for all $y_i,y_j\in\dom$} \\
        G_{\rho_i,\rho_j}(x_i,x_j)
                    &= \sum \left(
                           \prod_{k\neq i,j} T_{c_k,x_k}
                       \right)F_\rho(x_1,\dots,x_n)
                    && \text{ for all $x_i,x_j\in\{0,1\}$} \\
        H(y'\!,y'') &= \sum_{y\in\dom} G(y'\!,y)G(y''\!,y)
                    &&\text{ for all $y'\!,y''\in\dom$} \\
        H_{\rho_i,\rho_i}(x'\!,x'')
                    &= \sum_{x\in\{0,1\}}
                           G_{\rho_i,\rho_j}(x'\!,x)G_{\rho_i,\rho_j}(x''\!,x)
                    && \text{ for all $x'\!,x''\in\{0,1\}$}
    \end{align*}
    where the first sum is over all $y_1, \dots, y_{i-1}, y_{i+1},
    \dots, y_{j-1}, y_{j+1}, \dots, y_n\in D$ and the second sum is
    over all $x_1, \dots, x_{i-1}, x_{i+1}, \dots, x_{j-1}, x_{j+1},
    \dots, x_n\in\{0,1\}$.

    Note that $M_{H_{\rho_i,\rho_i}} = M_{G_{\rho_i,\rho_j}}
    \transpose{M_{G_{\rho_i,\rho_j}}} = T^{-1} M_B
    \transpose{(T^{-1})} T^{-1} M_B^t \transpose{(T^{-1})}$ where
    $\transposesymb$ denotes transpose. Taking determinants and
    applying Lemma \ref{lem:bool_facts} part~(\ref{pt:binary}) this
    implies that $H_{\rho_i,\rho_i}$ is not of rank~$1$.  Also, since
    $T$ is strictly positive, the support of $G_{\rho_i,\rho_j}$ is
    the binary projection of the support of $F_\rho$ onto $i$ and $j$ which, by
    assumption, is not $\NEQ$ or $\EQ$.  Hence $H_{\rho_i,\rho_i}$ is
    strictly positive but not of rank~$1$.  Again using
    Lemma~\ref{lem:bool_facts} part~(\ref{pt:binary}) we see that $H$
    is a witness that $\{F\}\cup \uns D$ is not weakly log-modular.
\end{proof}

\begin{lemma}\label{lem:logmod_balanced}
    Every conservative weakly log-modular weighted constraint language
    is balanced.
\end{lemma}
\begin{proof}
    Let $\nfuns$ be a conservative weighted constraint language that
    is not balanced.  We will show that $\nfuns$ is not weakly
    log-modular.

    By the definition of balance, there is a function
    $F\in\fclone{\nfuns}$ of arity $n$ and a partition
    $\vecx=(\vecu,\vecv)$ of its $n$ variables, such that the matrix
    $F(\vecu,\vecv)$ is not of block-rank~$1$.  By
    Lemma~\ref{lem:bool_facts} part~(\ref{pt:br1_2by2}) there is a
    two-by-two submatrix $N=N_{F,\vecu,\vecu'\!,\vecv,\vecv'}$ that is
    not of block-rank~$1$.

    Construct an essentially pseudo-Boolean function $G$ from $F$ as
    follows.  For all $1\leq i\leq n$ let $U_i\in\fclone{\uns
      D}\subseteq\fclone{\nfuns}$ be the indicator function of
    $D^{i-1}\times D_i\times D^{n-i}\!$, where $D_i=\{u_i,u'_i\}$ for
    all $1\leq i\leq k$ and $D_i=\{v_i,v'_i\}$ for all $k< i\leq
    n$.  Let $G=F\prod_{i=1}^n U_i$.  Then $N_{G, \vecu, \vecu'\!,
      \vecv, \vecv'}=N$ is not of block-rank~$1$, and $G$ is essentially
    pseudo-Boolean.

    If the binary projection of the support of $G$ onto two indices
    $i,j$ is a generalised $\NEQ$ $\{(x_i,x_j),(y_i,y_j)\}$, construct
    the tuple
    $(G'\!,\rho(\vecu),\rho(\vecu'),\rho(\vecv),\rho(\vecv'))$ from
    $(G,\vecu,\vecu'\!,\vecv,\vecv')$ as follows. Let
    $\rho\colon\dom^{n}\to\dom^{n-1}$ be the projection operator
    sending $\vecx$ to $x_1,\dots,x_{i-1},x_{i+1},\dots,x_n$ and let
    $G'(\vecx)=\sum_{\rho(\vecx')=\vecx}G(\vecx')$ for all
    $\vecx\in\dom^{n-1}\!$.  Note that, for all $\vecx\in D^n\!$,
    $G(\vecx)\neq G'(\rho(\vecx))$ implies that $G(\vecx)=0$ because
    $G(\vecx)=0$ unless $x_i\neq x_j$.  Note that~$N$ has
    at least three non-zero entries by Lemma \ref{lem:bool_facts}
    part~(\ref{pt:binary}). So the corresponding three pairs out of
    $((\vecu, \vecv)_i, (\vecu, \vecv)_j)$, $((\vecu, \vecv')_i,
    (\vecu, \vecv')_j)$, $((\vecu'\!, \vecv)_i, (\vecu'\!, \vecv)_j)$, and
    $((\vecu'\!, \vecv')_i ,(\vecu'\!, \vecv')_j)$ must each be either
    $(x_i,x_j)$ or $(y_i,y_j)$. But then the fourth pair must also be
    $(x_i,x_j)$ or $(y_i,y_j)$, which implies that $N_{G'\!, \rho(\vecu),
      \rho(\vecu'), \rho(\vecv), \rho(\vecv')} = N$.  Also, $G'$ is
    essentially pseudo-Boolean, and $G'$ is obtained by summing the $i$'th
    variable, so $G'\in\fclone{G}$.

    Repeating this process if necessary, we obtain
    $(G'\!,\vecx,\vecx'\!,\vecy,\vecy')$ such that $G'$ is an
    essentially pseudo-Boolean function in $\fclone{\nfuns,\uns D} =
    \fclone{\nfuns}$ and none of the binary projections of the support
    of~$G'$ is a generalised $\NEQ$, and
    $N_{G'\!,\vecx,\vecx'\!,\vecy,\vecy'}$ is not of block-rank~$1$.  So, in
    particular, $G'$ is not of rank~$1$. By
    Lemma~\ref{lem:balance_wmod_helper}, $\{G'\} \cup \uns D$ is not
    weakly log-modular, so $\fclone{\nfuns}$ is not weakly log-modular.
\end{proof}

We now return to Theorem~\ref{thm:main} and prove the tractable case.
The proof relies on an important theorem of Cai, Chen and Lu about the
complexity of exact evaluation.

\begin{lemma}\cite{NonNegExact}
\label{lem:NonNegExact_dichotomy}
    Let $\nfuns$ be a finite, weighted constraint language taking
    non-negative algebraic real values. If $\nfuns$ is balanced, then
    $\nCSP(\nfuns)$ is in \FP{}, and otherwise $\nCSP(\nfuns)$ is
    \numP{}-hard.
\end{lemma}

\begin{theorem}
\label{thm:tractable}
    Let $\nfuns$ be a conservative weighted constraint language taking
    values in $\posQ$.  If $\nfuns$ is weakly log-modular then, for
    any finite $\nfinfuns\subset\nfuns$, $\nCSP(\nfinfuns)\in \FP$.
\end{theorem}
\begin{proof}
    By Lemma \ref{lem:logmod_balanced}, $\nfuns$ is balanced. Hence,
    every finite $\nfinfuns\subset\nfuns$ is balanced, which implies
    that $\nCSP(\nfinfuns)$ is in \FP{} by Lemma
    \ref{lem:NonNegExact_dichotomy}.
\end{proof}

\section{Valued clones, valued CSPs and relational clones}
\label{sec:vcsp}
 
To define valued clones, we use the same set-up as
Section~\ref{sec:prelims} except that summation is replaced by
minimisation and product is replaced by sum.  Let $D$ be a finite
domain with $|D| \geq 2$ and let $R$ be a codomain with
$\{0,\infty\}\subseteq R$, where $\infty$ obeys the following rules
for all $x\in R$: $x+\infty=\infty$, $x\leq\infty$ and
$\min\{x,\infty\}=x$.  Let 
$\vfuns$ be a subset of $\funs R$ and let $V=\{v_1,\dots,v_n\}$ be a set
of variables.  For each atomic formula $\phi=G(v_{i_1},\dots,v_{i_a})$
we use the notation $f_\phi$ to denote the function represented
by~$\phi$, so $f_\phi(\vecx) = G(x_{i_1},\dots,x_{i_a})$.

A psm-formula (``primitive sum minimisation formula'') is a
minimisation of a sum of atomic formulas.  A psm-formula~$\psi$
over~$\vfuns$ in variables $V'=\{v_1,\dots,v_{n+k}\}$ has the form
\begin{equation} 
    \psi=\min_{v_{n+1},\dots,v_{n+k}}\,\sum_{j=1}^m\phi_j,
\end{equation}
where $\phi_j$ are all atomic formulas over~$\vfuns$ in the
variables~$V'\!$.  The formula~$\psi$ specifies a function
$f_\psi\colon D^n\to R$ in the following way:
\begin{equation} 
    f_\psi(\vecx)=\min_{\vecy\in\dom^k}\sum_{j=1}^m f_{\phi_j}(\vecx,\vecy),
\end{equation}
where $\vecx$ and $\vecy$ are assignments $\vecx\colon \{v_1, \dots,
v_n\}\to \dom$ and $\vecy\colon \{v_{n+1},\dots,\allowbreak v_{n+k}\}
\to \dom$.
 
The \emph{valued clone~$\vclone\vfuns$ generated by~$\vfuns$} is the
set of all functions that can be represented by a psm-formula
over~$\vfuns\cup\{\eq\}$, where $\eq$ is the binary equality function
on $\dom$ given by $\eq(x,x)=0$ and $\eq(x,y)=\infty$ for $x\neq y$.

We next introduce valued constraint satisfaction problems (\VCSP{}s),
which are optimisation problems.  In the work of Kolmogorov and
\Zivny{}~\cite{KZ}, the codomain is $R = \posQ \cup \{\infty\}$.  For
reasons which will be clear below, it is useful for us to extend the
codomain to include irrational numbers.  This will not cause problems
because, with the exception of Theorem~\ref{thm:STPMJN-tractable} we
use only formal calculations from their papers, not complexity results.  For
Theorem~\ref{thm:STPMJN-tractable}, we avoid irrational numbers and,
in fact, restrict to cost functions taking values in
$\{0,\infty\}\subset R$.  Furthermore, all the real numbers we use are
either rationals or their logarithms so are efficiently computable.

Let $\extR = {\mathbb R}_{\geq 0} \cup \{\infty\}$ be the set of
non-negative real numbers together with $\infty$.

\begin{definition}
    A \emph{cost function} is a function $D^k\to\extR$.  A
    \emph{valued constraint language} is a set of cost functions
    $\vfuns\subseteq\funs{\extR}$.
\end{definition}

Given a valued constraint language $\vfuns$, $\VCSP(\vfuns)$ is the
problem of taking an instance $\psi$, a psm-formula consisting of a
sum of $m$ atomic $\vfuns$-formulas over $n$ free variables $\vecx$
and computing the value
\begin{equation*}
    \mincost(\psi) = \min_{\vecx\in D^n } f_\psi(\vecx)\,,
\end{equation*}
where $f_\psi$ is the function defined by~$\psi$.
 
We typically use the notation of Kolmogorov and \Zivny.  An instance
is usually denoted by the letter $I$.  In this case, we use $f_I$ to
denote the function specified by the psm-formula corresponding to
instance~$I$, so the value of the instance is denoted
by~$\mincost(I)$.  The psm-formula corresponding to~$I$ is a sum of
atomic formulas (since all of the variables are free variables).  We
refer to each of these atomic formulas as a \emph{valued constraint}
and we represent these by the multiset $T$ of all valued constraints
in the instance~$I$.  For each valued constraint $t\in T$ we use $k_t$
to denote its arity, $f_t$ to denote the function represented by the
corresponding atomic formula, and $\sigma_t$ to denote its scope,
which is given as a tuple $(i(t,1),\ldots,i(t,k_t))\in \{1, \dots, n\}^{k_t}$
containing the indices of the variables in the scope. Thus,
\begin{equation}
\label{eq:instance}
    f_I(\vecx)=\sum_{t \in T} f_t(x_{i(t,1)},\dots,x_{i(t,k_t)})\,.
\end{equation}
For convenience, we use $\vecx[\sigma_t]$ as an abbreviation for the
tuple $(x_{i(t,1)},\dots,x_{i(t,k_t)})$.  In this abbreviated
notation, the function defined by instance~$I$ may be written
$f_I(\vecx)=\sum_{t \in T} f_t(\vecx[\sigma_t])$.

Now, let $\bddQ=[0,1]\cap{\mathbb{Q}}$.  For reasons which will be
clear below, it will be useful to work with weight functions in $\funs
\bddQ$. For such a weight function~$F$, let the cost function $\toVal
F\in \funs\extR$ be the function defined by
\begin{equation*}
     (\toVal F)(\vecx) = \begin{cases}
         -\ln F(\vecx) &\text{ if $F(\vecx)>0$}\\
         \infty        &\text{ if $F(\vecx)=0$.}
     \end{cases}
\end{equation*}
For example, $\toVal \EQ=\eq$,
where $\EQ$ and $\eq$ are
the functions defined earlier.
(Often, as here, we use a lower-case name like $\eq$ and 
an upper case name like $\EQ$
to indicate such a relationship.)
For a weighted constraint language
$\nfuns\subseteq \funs{ \bddQ}$, let $\toVal\nfuns$ be the valued
constraint language defined by $\toVal\nfuns = \{ \toVal F \mid F \in
\nfuns\}$.

There is a natural bijection between instances of $\nCSP(\nfuns)$ and
$\VCSP(\toVal\nfuns)$, obtained by replacing each function $F_t$ in
the former by the function $f_t=\toVal{F_t}$ in the latter, keeping
the scopes unchanged.  Note that $f_I(\vecx)=-\ln F_I(\vecx)$, for any
assignment~$\vecx$, with the convention $-\ln0=\infty$.
 
\begin{definition}
\label{def:conservativevalued} 
    A valued constraint language is \emph{conservative} if
    it contains all arity-1 cost functions $D\to\extR$.
\end{definition}

The mapping $F\mapsto \toVal F$ from $\funs{ \bddQ}$ to $\funs{
  \extR}$ is not surjective because there are real numbers that are
not the logarithm of any rational.  For the same reason, the valued
constraint language $\toVal \nfuns$ is not conservative (for any
weighted constraint language $\nfuns$).  Finally, note that we have
only defined $\toVal F$ for $F\in\funs{ \bddQ}$.  The obvious
extension to $F\in\funs{\posQ}$ would produce negative-valued cost
functions and we wish to avoid this since Kolmogorov and
\Zivny{}~\cite{KZ} do not allow it.
   
\begin{definition}  
    A cost function is \emph{crisp} \cite{CCJK2003:Soft} if
    $f(\vecx)\in\{0,\infty\}$ for all~$\vecx$.
\end{definition}

\begin{definition}
    For any cost function~$f$, let $\toRel f$ be the relation defined
    by $\toRel f = \{\vecx \mid f(\vecx)<\infty\}$.
\end{definition}

Thus, any cost function~$f$ can be associated with its underlying
relation.  Similarly, we can represent any relation by a crisp cost
function $f$ for which $f(\vecx)=0$ if and only if $\vecx$ is in the
relation.  A \emph{crisp constraint language} is a set of relations,
which we always represent as crisp cost functions, not as functions
with codomain $\{0,1\}$.  For a valued constraint language~$\vfuns$,
the crisp constraint language $\toRel \vfuns$ is given by $\toRel
\vfuns = \{ \toRel f \mid f\in \vfuns\}$.

\begin{definition}
\label{def:conservativelang}
    A crisp constraint language is \emph{conservative} if it
    includes all arity-1 relations.
\end{definition}

A \emph{relational clone} is simply a crisp constraint language $\toRel
{\vclone \vfuns}$ for a valued constraint language $\vfuns$.
 
\begin{lemma}\label{lem:closures}
    Suppose $\vfuns\subseteq\funs\extR$.  Then
    $\vclone{\toRel\vfuns}=\toRel{\vclone\vfuns}$.
\end{lemma}
\begin{proof}
    The mapping $\rho\colon\extR\to\{0,\infty\}$ defined by
    $\rho(\infty)=\infty$ and $\rho(x)=0$, for all $x<\infty$, is a
    homomorphism of semirings, from $(\extR,\min,+)$ to
    $(\{0,\infty\},\min,+)$.
\end{proof}

\section{STP/MJN multimorphisms and weak log-supermodularity}\label{sec:stp}

In \cite[Corollary~3.5]{KZ}, Kolmogorov and \Zivny\ give a tractability
criterion for conservative \VCSP{}s.  In particular, they show that
the \VCSP{} associated with a conservative valued constraint
language~$\vfuns$ is tractable iff $\vfuns$ has an STP/MJN
multimorphism.

We define STP/MJN multimorphisms below.  In this section, we show
(see Theorem~\ref{thm:wlsm_to_stp} below) that, if a weighted
constraint language~$\nfuns\in\funs{\bddQ}$ is weakly
log-supermodular, then the corresponding valued constraint language
$\toVal \nfuns$ has an STP/MJN multimorphism. In 
Section~\ref{sec:lsm_bis_easy}, this will enable us to use such a
multimorphism
(via the work of Kolmogorov and \Zivny \cite{KZ} and Cohen, Cooper and
Jeavons \cite{CCJ}) to prove \nBIS{}-easiness and \LSM{}-easiness of the
weighted counting CSP.
 
Our proof of Theorem~\ref{thm:wlsm_to_stp} relies on work by
Kolmogorov and \Zivny{} \cite{KZ} and
Takhanov~\cite{TakhanovFull}.  We start with some general definitions.
Most of these are from~\cite{KZ}, but some care is
required since some of the definitions in~\cite{KZ} differ from those in
\cite{CCJ}.

\begin{definition}
    A \emph{$k$-ary operation on~$D$} is a function from $D^k$ to $D$.
    An \emph{operation on~$D$} is a $k$-ary operation, for some $k$.
\end{definition}

We drop the ``on~$D$'' when the domain~$D$ is clear from the context.

\begin{definition}
    A $k$-tuple $\langle\rho_1, \dots, \rho_k\rangle$ of $k$-ary
    operations $\rho_1, \dots, \rho_k$ is \emph{conservative} if, for
    every tuple $\vecx=(x_1, \dots, x_k)\in D^k\!$, the multisets $\{\{x_1,
    \dots, x_k\}\}$ and $\{\{\rho_1(\vecx), \dots, \rho_k(\vecx)\}\}$
    are equal.
\end{definition}

Note that we have now defined conservative operations and conservative
constraint languages (weighted, valued and crisp).  There are
connections between these notions of ``conservative'' but we do not
need these here.

\begin{definition}\label{def:multimorphism}
    $\langle\rho_1,\dots,\rho_k\rangle$ is a
    \emph{multimorphism} of an arity-$r$ cost function~$f$ if, for all
    $\vecx^1\!,\dots,\vecx^k\in D^r\!$, we have:
    \begin{equation*}
        \sum_{i=1}^k f(\rho_i(x_1^1,\ldots,x_1^k), \dots,
                      \rho_i(x_r^1,\ldots,x_r^k)) 
        \leq \sum_{i=1}^k f(\vecx^i).
    \end{equation*}
\end{definition}

\begin{definition} \label{def:multimorphismvfuns}
    $\langle\rho_1,\dots,\rho_k\rangle$ is a multimorphism of a valued
    constraint language $\vfuns$ if it is a multimorphism of every
    $f\in\vfuns$.
\end{definition}

These definitions imply the following.

\begin{observation}
\label{obs:unarymultimorphism}
    If $\langle\rho_1,\dots,\rho_k\rangle$ is conservative, then it is
    a multimorphism of every unary cost function~$f$.
\end{observation}
 
\begin{definition} 
    Suppose $M \subseteq D^2\!$.  A pair $\langle\sqcap,\sqcup\rangle$
    of binary operations is a \emph{symmetric tournament pair} (STP)
    on~$M$ if it is conservative and both operations are commutative
    on~$M$. We say that it is an STP if it is an STP on~$D^2\!$.
\end{definition}

\begin{definition}
\label{def:MMJJNN}
    Suppose $M \subseteq D^2\!$.  A triple $\langle\mathtt{Mj1}, \mathtt{Mj2},
    \mathtt{Mn3}\rangle$ of ternary operations is an \emph{MJN
      on $ M$} if it is conservative and, for all triples $(a, b, c)
    \in D^3$ with $\{\{a, b, c\}\} = \{\{x, x, y\}\}$ where $x$ and
    $y$ are distinct and $(x,y)\in M$, we have $\mathtt{Mj1}(a, b,
    c)=\mathtt{Mj2}(a, b, c)=x$ and $\mathtt{Mn3}(a, b, c)=y$.
\end{definition}
 
The reason that Definition~\ref{def:MMJJNN}
only deals with the case in which~$x$ and~$y$ are distinct
is that any conservative triple 
$\langle\mathtt{Mj1}, \mathtt{Mj2},
    \mathtt{Mn3}\rangle$ 
    satisfies
    $\mathtt{Mj1}(x,x,x) = 
    \mathtt{Mj2}(x,x,x) = 
    \mathtt{Mn3}(x,x,x) = x$.

\begin{definition} \label{def:STPMJN}
    An \emph{STP/MJN multimorphism} of a valued constraint language~$\vfuns$
    consists of a pair of operations $\langle \sqcap,\sqcup \rangle$
    and a triple of operations $\langle\mathtt{Mj1}, \mathtt{Mj2},
    \mathtt{Mn3} \rangle$, both of which are multimorphisms
    of~$\vfuns$, for which, for some symmetric subset $M$ of~$D^2\!$,
    $\langle\sqcap,\sqcup\rangle$ is an STP on~$M$ and
    $\langle\mathtt{Mj1}, \mathtt{Mj2}, \mathtt{Mn3}\rangle$ is an MJN
    on $\{(a,b)\in D^2\setminus M \mid a\neq b\}$.
\end{definition}

\begin{definition} \label{def:weaklysubmod}
    $\vfuns\subseteq\funs\extR$ is \emph{weakly submodular} if, for
    all binary functions $f\in\vclone\vfuns$ and elements $a,b\in D$,
    \begin{equation}\label{eq:softloopVal}
        f(a,a)+f(b,b)\leq f(a,b)+f(b,a)
        \quad\text{or}\quad
        f(a,a)=f(b,b)=\infty.
    \end{equation}
\end{definition}

Note that the definition of weak submodularity for cost functions is a
restatement of Kolmogorov and \Zivny's ``Assumption 3''.  It is not
trivial that weak log-supermodularity for $\nfuns$ is related to weak
submodularity for $\toVal\nfuns$. Expressibility for \VCSP{} is
different from expressibility for \nCSP{} and, specifically, we cannot
expect $\vclone{\toVal\nfuns} = \toVal{\fclone\nfuns}$ to hold in
general.  However, the following is suitable for our purposes.

\begin{lemma}\label{lem:softlooptranslate}
    Suppose $\nfuns\subseteq\funs{\bddQ}$ and let
    $\vfuns=\toVal\nfuns$.  If $\nfuns$ is weakly log-supermodular
    then $\vfuns$ is weakly submodular.
\end{lemma}
\begin{proof} 
    We prove the contrapositive.  Suppose $f\in\vclone\vfuns$ is a
    binary function that witnesses the fact that $\vfuns$ is not
    weakly submodular according to Definition~\ref{def:weaklysubmod},
    specifically,
    \begin{equation*}
        f(a,a)+f(b,b)>f(a,b)+f(b,a)
        \quad\text{and}\quad
        \min\{f(a,a),f(b,b)\}<\infty.
    \end{equation*}
    Since $f\in\vclone\vfuns$, we may express $f$ in the form
    \begin{equation*}
        f(\vecx) = \min_\vecy g(\vecx,\vecy)
            = \min_\vecy\sum_{i=1}^m g_i(\vecx,\vecy),
    \end{equation*}
    where the $g_i\in\vfuns$ are atomic.  For $k\in\nset$, define
    \begin{equation*}
        F^{(k)}(\vecx)=\sum_\vecy\prod_{i=1}^m G_i(\vecx,\vecy)^k,
    \end{equation*}
    where each $G_i$ is such that $g_i=\toVal{G_i}$.  Note that
    $F^{(k)}\in\fclone{\nfuns}$, and
    \begin{equation*}
        F^{(k)}(\vecx)^{1/k}\>\to\>\max_\vecy\prod_{i=1}^m
        G_i(\vecx,\vecy),\quad\text{as $k\to\infty$}.
    \end{equation*}
    Now
    \begin{align*}
        \max_\vecy\prod_{i=1}^m G_i(\vecx,\vecy)
            &=\max_\vecy\exp\bigg(-\sum_{i=1}^m g_i(\vecx,\vecy)\bigg)\\
            &=\exp\bigg(-\min_\vecy\sum_{i=1}^m g_i(\vecx,\vecy)\bigg)\\
            &=\exp(-f(\vecx))
    \end{align*}
    and
    \begin{equation*}
        \exp(-f(a,a))\exp(-f(b,b))<\exp(-f(a,b))\exp(-f(b,a)).
    \end{equation*}
    Thus $F(a,a)F(b,b)<F(a,b)F(b,a)$ where $F=F^{(k)}$ for some
    sufficiently large~$k$.  Also, $\min\{f(a,a),f(b,b)\}<\infty$
    implies that $\max\{F(a,a),F(b,b)\}>0$.  These properties of $F$
    imply that $\nfuns$ is not weakly log-supermodular, according to
    Definition~\ref{def:weaklylogsupermod}.
\end{proof}

Let $\rels$ be a crisp constraint language.  A \emph{majority polymorphism}
of $\rels$ is a ternary operation~$\rho$ such that
$\rho(a,a,b)=\rho(a,b,a)=\rho(b,a,a)=a$ for all $a,b\in D$ and for all
arity-$k$ relations $R\in\rels$ we have
\begin{equation*}
    \vecx,\vecy,\vecz\in R \implies
        (\rho(x_1,y_1,z_1),\dots,\rho(x_k,y_k,z_k))\in R.
\end{equation*}

Let $N(a,b,c,d)$ be the relation $\{(a,c), (b,c), (a,d)\}$.  The existence of
such a relation in $\vclone{\rels}$ indicates that $\rels$ is not
``strongly balanced'' in the terminology of \cite{DReffective}.  Note
that, on the Boolean domain, $N(0,1,0,1)$ is the ``NAND'' relation.

\begin{theorem}(Takhanov)
\label{thm:Takhanov}
    Let $\rels$ be a conservative relational clone with domain $D$.
    At least one of the following holds.
    \begin{itemize}
    \item There are distinct $a,b\in D$ such that $N(a,b,a,b)\in
        \rels\!$.
    \item There are distinct $a,b\in D$ such that $\{(a,a,a), (a,b,b),
        (b,a,b),(b,b,a)\} \in \rels\!$.
    \item For some $k\geq 1$, there are $a_0, \dots, a_{2k}, b_0,
        \dots, b_{2k}\in D$ such that, for each $0\leq i\leq 2k$,
        $a_i\neq b_i$ and, for each $0\leq i\leq 2k-1$,
        \begin{equation*}
            N(a_i, b_i, a_{i+1}, b_{i+1}) \in \rels
            \text{ and }
            N(a_{2k}, b_{2k}, a_0, b_0) \in \rels.
        \end{equation*}
    \item $\rels$  has a majority polymorphism.
    \end{itemize}
\end{theorem}
\begin{proof}
    This formulation is essentially \cite[Theorem~9.1]{TakhanovFull}
    except for the last bullet point.  As stated in the proof of that
    theorem, the first two conditions both fail if and only if the
    ``necessary local conditions'' \cite[Definition 3.5]{TakhanovFull}
    hold.  Unfortunately for us, Takhanov uses the term ``functional
    clone'' differently to how we use it, so the reader will need to
    take this into account to understand the local conditions.
    However, we do not need the detail, here.
    Takhanov's proof of the \NP{}-hard case of his
    Theorem~3.7 (at the end of his Section~4) shows the following:
    Given the necessary local conditions, the third condition fails
    only if a certain graph $T_F$ is bipartite. If $T_F$ is
    bipartite then \cite[Theorem 5.5]{TakhanovFull} establishes a
    majority polymorphism.
\end{proof}

\begin{lemma}\label{lem:existsMaj}
    If $\vfuns\subseteq\funs\extR$ is conservative and weakly
    submodular, $\rels=\vclone{\toRel\vfuns}$ has a majority
    polymorphism.
\end{lemma}
\begin{proof}
    Since $\vfuns$ is conservative
    (Definition~\ref{def:conservativevalued}), so is~$\rels$
    (Definition~\ref{def:conservativelang}).  We will now show that
    the first three bullets of Theorem~\ref{thm:Takhanov} contradict
    the premise of the lemma, so the fourth must hold.

    The first bullet-point is easily ruled out.  Suppose the given
    relation is in~$\rels\!$. By Lemma~\ref{lem:closures}, there is a
    binary function $f\in\vclone\vfuns$ such that $\toRel f =
    N(a,b,a,b)$.  This function has $f(b,b)=\infty$ and
    $f(a,a),f(a,b),f(b,a)<\infty$, and hence
    violates~(\ref{eq:softloopVal}).

    For the second bullet-point, by Lemma~\ref{lem:closures} we have  
            an arity-$3$
    function $g\in\vclone\vfuns$ which is finite precisely on
    $\big\{(a,a,a),(a,b,b),(b,a,b),(b,b,a)\big\}$.  Now let $M$ be a
    sufficiently large constant and let $u$ be the unary function
    defined by
    \begin{equation*}
        u(z) = \begin{cases}
                   M      & \text{if $z=a$,}\\
                   0      & \text{if $z=b$,}\\
                   \infty & \text{otherwise.}
        \end{cases}
    \end{equation*}
    Let
    \begin{equation*}
        f(x,y)=\min_{z\in D}\{g(x,y,z)+u(z)\}.
    \end{equation*}
    Then $f(a,a)=M+g(a,a,a)$, $f(b,b)=M+g(b,b,a)$, $f(a,b)=g(a,b,b)$
    and $f(b,a)=g(b,a,b)$.  Clearly, $f$
    violates~(\ref{eq:softloopVal}) for sufficiently large~$M$.
 
    Finally, let us consider the third bullet-point.  By
    Lemma~\ref{lem:closures} we have binary functions
    $g_0,g_1,\dots,     g_{2k}
    \in\vclone\vfuns$ where the underlying
    relation of $g_i$ is $N(a_i, b_i, a_{i+1}, b_{i+1})$ for $0\leq i <
    2k$ and the underlying relation of $g_{2k}$ is $N(a_{2k}, b_{2k},
    a_0, b_0)$.  Define
    \begin{multline*}
        f(x,y)= \min\big\{g_0(x,z_0) + u_0(z_0) + g_1(z_0,z_1)
                          + u_1(z_1)+\null                          \\
                          \cdots + u_{2k-1}(z_{2k-1})
                          + g_{2k}(z_{2k-1},y)
                          \mid (z_0,\dots,z_{2k-1})\in D^{2k}
                     \big\},
    \end{multline*}
    where $u_i(a_i)=M$, $u_i(b_i)=0$, and $u_i(z)=\infty$ if
    $z\notin\{a_i,b_i\}$.  Note that $f(a_0,a_0) \geq kM$, $f(b_0,b_0) \geq
    (k+1)M$ and $f(a_0,b_0), f(b_0,a_0)\leq kM+(2k+1)m$, where $m$ is the
    largest finite value taken by any of $g_0, \dots, g_{2k}$.  So
    $f$ violates~(\ref{eq:softloopVal}) for sufficiently large~$M$.

    So we are left with the remaining possibility that $\rels$ has a
    majority polymorphism.
\end{proof}

We can now prove the main result of this section.  

\begin{theorem}\label{thm:wlsm_to_stp}
    Let $\nfuns$ be a weighted constraint language such that
    $\func_1(D,\bddQ) \subseteq \nfuns \subseteq\funs{\bddQ}$.
    If $\nfuns$ is weakly log-supermodular then $\toVal{\nfuns}$ has
    an STP/MJN multimorphism.
\end{theorem}
\begin{proof}
    Let $\vfuns=\toVal\nfuns\cup\func_1(D,\extR)$. We will show that
    $\vfuns$ has an STP/MJN multimorphism.  By
    Definitions \ref{def:STPMJN} and~\ref{def:multimorphismvfuns},
    this is also an STP/MJN multimorphism of the subset
    $\toVal\nfuns$.
  
    By Lemma~\ref{lem:softlooptranslate}, $\toVal\nfuns$ is weakly
    submodular.  Now, $\toVal \nfuns$ contains $\toVal
    {\func_1(D,\bddQ)}$.  Thus, for every unary function $u\in
    \func_1(D,\extR)$ and every $\epsilon\in(0,1)$, there is a unary
    function $u_\epsilon\in \toVal \nfuns$ such that, for all
    $x\in\{0,1\}$, $|u(x)-u_\epsilon(x)|<\epsilon$.  From the
    definition of valued clones, and continuity, we deduce that, for
    every binary function $f\in \vclone \vfuns$ and every
    $\epsilon>(0,1)$, there is an $f_\epsilon \in \vclone {\toVal \nfuns}$
    such that, for all $x,y\in\{0,1\}$, $|f(x,y) - f_\epsilon(x,y)| <
    \epsilon$.  Since $\toVal\nfuns$ is weakly submodular, we conclude
    from the definition of weak submodularity
    (Definition~\ref{def:weaklysubmod}) that $\vfuns$ is weakly
    submodular.

    In \cite[\S6.1--6.4]{KZ}, Kolmogorov and \Zivny{} show how to
    construct an STP/MJN multimorphism of~$\vfuns$ under ``Assumptions
    1--3''.  Assumption~1 is that $\vfuns$ is conservative, which is
    true by construction.  Assumption~3 is that $\vfuns$ is weakly
    submodular.  This is given as a premise of our lemma.
    Assumption~2 is that $\rels=\toRel\vfuns$ has a majority
    polymorphism, which follows from Assumptions 1 and~3 by
    Lemma~\ref{lem:existsMaj}.  (Assumption~2 states that $\vfuns$ has
    a majority polymorphism. In our terminology, this means that
    $\toRel\vfuns$ has a majority polymorphism.)
\end{proof}

\section{\LSM{}-easiness and \nBIS{}-easiness}\label{sec:lsm_bis_easy}

Our aim is to show that if $\toVal\nfuns$ has an STP/MJN multimorphism
then $\nfuns$ is \LSM{}-easy.  This will involve using the arguments
of \cite{CCJ} and \cite{KZ}, but we try, as much as possible, to avoid
going into the details of their proofs.  We start by generalising the
notion of an STP multimorphism.

\begin{definition}\label{def:genSTP}
    Let $f$ be an arity-$k$ cost function.  A \emph{multisorted
    multimorphism of~$f$} is a pair $\langle\sqcap,\sqcup\rangle$,
    defined as follows.  For $1\leq i\leq k$, $\sqcap_i$ and $\sqcup_i$
    are operations on the set $D_i = \{a \in D \mid \exists \vecx:
    x_i=a \text{ and } f(\vecx)<\infty\}$, and $\langle\sqcap_i,
    \sqcup_i\rangle$ is an STP of $\{f\}$. 

    The operation~$\sqcap$ is the binary operation on $D_1\times
    \dots \times D_k$ defined by applying $\sqcap_1,\ldots,\sqcap_k$
    component-wise.  Similarly, $\sqcup$ is defined by applying
    $\sqcup_1,\ldots,\sqcup_k$ component-wise.  We require that, for
    all $\vecx,\vecy\in D^k\!$, $ f(\sqcup(\vecx,\vecy)) +
    f(\sqcap(\vecx,\vecy)) \leq f(\vecx) + f(\vecy)$.  Equivalently,
    we require
    \begin{equation*}
        f( \sqcup_1(x_1,y_1), \ldots, \sqcup_k(x_k,y_k))
            + f( \sqcap_1(x_1,y_1), \ldots, \sqcap_k(x_k,y_k))
        \leq f(\vecx) + f(\vecy).
    \end{equation*}
\end{definition}

Kolmogorov and \Zivny~\cite[Equation (35)]{KZ} use a slightly more general definition, where
$D_i$ can be any superset of $\{a \in D \mid \exists \vecx: x_i=a
\text{ and } f(\vecx)<\infty\}$.  But it does no harm to be more
specific.

Where it is clearer, we use infix notation for operations such as
$\sqcap$ and~$\sqcup$.

\begin{theorem}[Kolmogorov and \Zivny] \label{thm:kzeasy} 
    Suppose $\vfuns_0$ is a finite, valued constraint language which has an
    STP/MJN multimorphism.  Then there is a polynomial-time algorithm
    that takes an instance~$I$ of $\VCSP(\vfuns_0)$ and returns a
    multisorted multimorphism $\langle\sqcap,\sqcup\rangle$ of
    $f_I$.  The pair $\langle\sqcap,\sqcup\rangle$ depends only on the
    STP/MJN multimorphism of~$\vfuns_0$ and on the relation
    $\toRel{f_I}$ underlying~$f_I$. It does not depend in any other
    way on~$I$.
\end{theorem}
\begin{proof} 
    This is proved by Stages $1$ and~$2$ of the proof of Theorem~3.4
    in~\cite[\S7]{KZ}, in which Kolmogorov and \Zivny{} establish the
    existence of the pair $\langle\sqcap,\sqcup\rangle$ that we
    require.

    Note that Kolmogorov and \Zivny{} restrict to rationals, whereas we 
    allow real numbers, but this is not a problem.  Their proof
    constructs $\langle\sqcap,\sqcup\rangle$ using an algorithm but
    this algorithm
    does not require access to the functions in~$\vfuns_0$ themselves.
    Instead, it only requires access to the relations in
    $\toRel{\vfuns_0}$ and to the STP/MJN multimorphism
    that~$\vfuns_0$ satisfies.  These are both finite amounts of data,
    which can be hardwired into the algorithm, whose input is just the
    instance $I$, which is a symbolic expression.
\end{proof}

We will also use the following algorithmic consequence of
\cite[Theorem~3.4]{KZ}.  We restrict to crisp cost functions because
this is all that we use and we wish to avoid issues with number
systems.

\begin{theorem}[Kolmogorov and \Zivny]
\label{thm:STPMJN-tractable}
Suppose that $\vfuns_0$ is a finite, crisp constraint language that has an
STP/MJN multimorphism. Then there is a polynomial-time algorithm for 
$\VCSP(\vfuns_0)$.
\end{theorem}

For our eventual construction, we would like
$\langle\sqcap,\sqcup\rangle$ to induce a multisorted
multimorphism of~$f_t$ for each individual valued constraint~$t$ in
the instance.  We do not know whether this is true of the multisorted
multimorphism provided by Kolmogorov and \Zivny's algorithm, but
something sufficiently close to this is true.

\begin{definition}
    For an instance~$I$, a valued constraint~$t$ and a length-$k_t$
    vector~$\veca$, define
    \begin{equation*}
        R_{I,t}(\veca)=\begin{cases}
            0,      &\text{if there exists $\vecx$ with
                      $\vecx[\sigma_t]=\veca$ and $f_I(\vecx)<\infty$};\\
            \infty, &\text{otherwise},
        \end{cases}
    \end{equation*}
    and define $f'_t=f_t+R_{I,t}$.  
\end{definition}

Thus, $f'_t$ is a ``trimmed'' version of $f_t$ whose domain is
precisely the $k_t$-tuples of values that can actually arise in
feasible solutions to instance~$I$.  We will see that if the scope
$\sigma_t$ contains variables with indices $i(t,1),\ldots,i(t,k_t)$,
then
\begin{equation*}
    \big\langle{\sqcap}[\sigma_t],\sqcup[\sigma_t]\big\rangle=
    \big\langle(\sqcap_{i(t,1)},\dots,\sqcap_{i(t,k_t)}),(\sqcup_{i(t,1)},\dots,\sqcup_{i(t,k_t)})\big\rangle
\end{equation*}
is a multisorted multimorphism of~$f'_t$, even though it might not
necessarily be a multisorted multimorphism of~$f_t$.

Note that Theorem~\ref{thm:STPMJN-tractable} has the following
consequence.

\begin{corollary}
\label{cor:STPMJN-tractable}
    Let $\vfuns_0$ be a finite, valued constraint language that has an
    STP/MJN multimorphism.  There is a polynomial-time algorithm that
    takes an instance $I$ of $\VCSP(\vfuns_0)$, a valued constraint $t$
    and 
    returns a truth table for $f'_t$.
\end{corollary}
\begin{proof}
    By Observation~\ref{obs:unarymultimorphism}, any STP/MJN
    multimorphism of $\vfuns_0$ is also an STP/MJN multimorphism of
    $\vfuns_0' = \vfuns_0\cup \func_1(\dom,\{0,\infty\})$.  Now, for
    each vector $\veca\in\dom^{k_t}$ in turn, we can determine the
    value of $f'_t(\veca)$ as follows.  Let $I_\veca$ be the
    $\VCSP(\vfuns_0')$ instance that results from adding to~$I$
    the set of~$k_t$ crisp, unary, valued constraints that force the
    tuple of variables $\vecx[\sigma_t]$ to take value~$\veca$.  By
    Theorem~\ref{thm:STPMJN-tractable}, we can compute in polynomial
    time whether $f_{I_\veca}(\veca) < \infty$ and, thus, determine
    the value of~$f'_t(\veca)$.
\end{proof}

The truth table produced by the algorithm of
Corollary~\ref{cor:STPMJN-tractable} is finite since all valued
constraints in $\vfuns_0$ are finite.

\begin{theorem}[An extension to Theorem~\ref{thm:kzeasy}]\label{thm:extendkzeasy} 
    Suppose $\vfuns_0$ is a finite, valued constraint language which has an
    STP/MJN multimorphism.  Consider the algorithm from
    Theorem~\ref{thm:kzeasy} which takes an instance~$I$ of
    $\VCSP(\vfuns_0)$ (in the form~(\ref{eq:instance})) and returns a
    multisorted multimorphism $\langle\sqcap,\sqcup\rangle$ of
    $f_I$.  Then, for all $t\in T$, $\big\langle{\sqcap}[\sigma_t],
    \sqcup[\sigma_t]\big\rangle$ is a multisorted multimorphism of
    $f'_t$.
\end{theorem}
\begin{proof}
    Focus on a particular valued constraint $t$ of~$I$.  Let $k=k_t$
    be the arity of $f_t$, and for brevity denote $\sqcap[\sigma_t]$
    and $\sqcup[\sigma_t]$ by $\sqcap'$ and $\sqcup'\!$, respectively.
    Without loss of generality assume $\sigma_t=(1,2,\dots, k)$.  We
    wish to show that
    \begin{equation}\label{eq:multif'}
        f'_t(\veca\sqcap'\vecb)+f'_t(\veca\sqcup'\vecb)
            \leq f'_t(\veca)+f'_t(\vecb)
    \end{equation}
    for all $\veca,\vecb\in D_1\times\dots\times D_k$.  If either
    $f'_t(\veca)=\infty$ or $f'_t(\vecb)=\infty$, then we are done.
    Otherwise, by construction of~$f_t'$, there exist $\vecx$ and
    $\vecy$ such that $\veca=\vecx[\sigma_t]$,
    $\vecb=\vecy[\sigma_t]$, and $f_I(\vecx),f_I(\vecy)<\infty$.
    Notice that $f'_t(\veca)=f_t(\veca)<\infty$ and
    $f'_t(\vecb)=f_t(\vecb)<\infty$, also by construction of~$f'_t$.
    Now consider the augmented instance~$I_N$ of~$I$ with $N$ extra
    copies of the valued constraint $t$.  We have
    \begin{equation}\label{eq:InversusI}
    \begin{split}
        f_{I_N}(\vecx)&=f_I(\vecx)+Nf_t(\veca)\\
        f_{I_N}(\vecy)&=f_I(\vecy)+Nf_t(\vecb)\\
        f_{I_N}(\vecx\sqcap \vecy)&=f_I(\vecx\sqcap\vecy)+Nf_t(\veca\sqcap'\vecb)\\
        f_{I_N}(\vecx\sqcup \vecy)&=f_I(\vecx\sqcup\vecy)+Nf_t(\veca\sqcup'\vecb).
    \end{split}
    \end{equation}
    Since $\toRel {f_{I_N}}=\toRel {f_I}$, from
    Theorem~\ref{thm:kzeasy}, $\langle\sqcap,\sqcup\rangle$ is also a
    multisorted multimorphism of $f_{I_N}$, i.e.,
    \begin{equation*}
        f_{I_N}(\vecx\sqcap \vecy) + f_{I_N}(\vecx\sqcup \vecy)
            \leq f_{I_N}(\vecx)+f_{I_N}(\vecy).
    \end{equation*}
    Combining this with~(\ref{eq:InversusI}), we obtain
    \begin{equation}
    \label{eq:limit}
        f_t(\veca\sqcap'\vecb)+f_t(\veca\sqcup'\vecb)+O(N^{-1})\leq
        f_t(\veca)+f_t(\vecb)+O(N^{-1}).
    \end{equation}
    Since \eqref{eq:limit} remains true as $N\to\infty$ but $f_t$ is
    independent of $N$, we conclude that
    \begin{equation*}
        f_t(\veca\sqcap'\vecb) + f_t(\veca\sqcup'\vecb)
            \leq f_t(\veca)+f_t(\vecb).
    \end{equation*}
    Since $\veca\sqcap'\vecb$ and $\veca\sqcup'\vecb$ extend to
    feasible solutions $\vecx\sqcap\vecy$ and $\vecx\sqcup\vecy$, it
    follows that $f'_t(\veca\sqcap'\vecb)=f_t(\veca\sqcap'\vecb)$ and
    $f'_t(\veca\sqcup'\vecb)=f_t(\veca\sqcup'\vecb)$.  The required
    inequality~(\ref{eq:multif'}) follows immediately.
\end{proof}

To make use of Theorem~\ref{thm:extendkzeasy}, we will use the
following definitions.

\begin{definition}
    Given a finite, valued constraint language $\vfuns_0\subset\funs\extR$, let
    $\vfuns_0'$ be the set of functions of the form $f+R$, for
    $f\in\vfuns_0\cap\funsk\extR$, $R\in\funsk{\{0,\infty\}}$ and
    $k\in\nset$.
\end{definition}

Note that $\vfuns_0'$ is finite because $\funsk{\{0,\infty\}}$ is
finite for any finite~$k$.

\begin{definition}
\label{def:equiv}
Suppose that $I$ is an $n$-variable instance of $\VCSP(\vfuns)$ with domain~$D$
and $I'$ is an $n'$-variable instance of $\VCSP(\vfuns')$ with domain~$D'$.
We say   
that $I$ and $I'$ are \emph{equivalent}
if there is a bijection~$\pi$
from  $\{ \vecx\in D^n \mid f_I(\vecx)<\infty\}$
to 
$\{ \vecx'\in {D'}^{n'} \mid f_{I'}(\vecx')<\infty\}$
such that, for all
$\vecx$ in the domain of~$\pi$,
$f_I(\vecx) = f_{I'}(\pi(\vecx))$. \end{definition}
 
Lemma~\ref{cor:equivInstExtn} below will construct equivalent
instances~$I$ and~$I'$ in a setting where $n=n'$ and $D=D'$.
In this case, $\pi$ will be the identity bijection from~$D^n$ to itself.
Later, we will consider equivalences of instances with different domains.

\begin{definition} \cite{CCJ}
    A function $f\colon D_1 \times \dots \times D_r \to \extR$ is
    \emph{domain-reduced} if, for each $i\in\{1, \dots, r\}$, and for
    each $a\in D_i$, there is an $\vecx\in \dom^n$ such that $x_i=a$
    and $f(\vecx)<\infty$.
\end{definition}

\begin{lemma}\label{cor:equivInstExtn}
    Suppose $\vfuns_0$ is a finite, valued constraint language which has an
    STP/MJN multimorphism.  Consider an instance~$I$ of
    $\VCSP(\vfuns_0)$.  There is an equivalent instance $I'$ of
    $\VCSP(\vfuns_0')$ and a multisorted multimorphism
    $\langle\sqcap,\sqcup\rangle$ of~${f_{I'}}$ which induces a
    multisorted multimorphism of~$f_t$ for each valued constraint
    $t$ of $I'\!$.  Both $I'$ and $\langle\sqcap,\sqcup\rangle$ are
    polynomial-time computable (given~$I$).  Moreover, each operation
    $\sqcap_i$ and $\sqcup_i$ induces a total order.
\end{lemma}
\begin{proof}
    We first show how to construct an equivalent instance~$I'$ and a
    multisorted multimorphism of $f_{I'}$ which induces
    multisorted multimorphisms on the valued constraints.  To
    obtain $I'\!$, start from the instance~$I$ and use
    Corollary~\ref{cor:STPMJN-tractable} to replace each valued
    constraint $f_t(\vecx[\sigma_t])$ with $f_t'(\vecx[\sigma_t])$.
    This operation clearly preserves the set of feasible solutions and
    their costs.  Then use the algorithm from
    Theorem~\ref{thm:extendkzeasy} to construct the multisorted
    multimorphism $\langle\sqcap,\sqcup\rangle$.

    In the remainder of the proof, we construct a new multisorted
    multimorphism by modifying $\langle\sqcap,\sqcup\rangle$ to ensure
    that  
its components induce total orders,
as required. Consider the
    following claim.

    {\sl Claim:\quad Suppose that D is a domain.  Given a set of
      functions~$\vfuns\subseteq \funs{\extR}$, let $\mathcal{P}$ be
      an instance of $\VCSP(\vfuns)$ with variable set
      $\{v_1,\ldots,v_n\}$.  Let $D_i = \{a \in D \mid \exists \vecx:
      \mbox{$x_i=a$ and $f_{\cal P}(\vecx)<\infty$}\}$.  Suppose that
      $\langle \sqcap,\sqcup\rangle$ is a multisorted
      multimorphism of~$\mathcal{P}$.  Then there is a multisorted
      multimorphism $\langle \sqcap'\!,\sqcup'\rangle$ of~$\mathcal{P}$
      in which each $\sqcap'_i$ 
induces      
      a total order on~$D_i$ (hence
      $\sqcup'_i$ 
induces
the reversal of this total order).  Furthermore,
      for any set $J=\{i_1,\ldots,i_j\} \subseteq \{1, \dots, n\}$ and any
      domain-reduced function $\oldphi\colon D_{i_1} \times \dots \times D_{i_j}
      \to \extR$ for which $\langle \sqcap_J,\sqcup_J\rangle$
      is a multimorphism, $\langle \sqcap'_J,\sqcup'_J\rangle$ is also
      a multimorphism of~$\oldphi$.  The multimorphism $\langle
      \sqcap'\!,\sqcup'\rangle$ is polynomial-time computable.}

    This claim is proved (but not explicitly stated) in the proof of
    \cite[Theorem~8.2]{CCJ}.\footnote{\cite{CCJ} uses somewhat
      different notation to ours: our $\sqcap$ and~$\sqcup$ are their
      $f$ and~$g$, respectively; in \cite{CCJ}, $\langle f,g\rangle$
      denotes an ordered pair, whereas we use that notation to denote
      a clone.}
    The basic method is as follows.
    $\mathcal{P}$ is augmented with extra (redundant) valued
    constraints using unary and binary crisp cost functions.  The
    binary crisp cost functions are used to enforce consistency so
    that when $\sqcap_i$ is modified 
   to induce
 a total order on~$D_i$, a
    compatible modification is made to each other $\sqcap_j$.  Once
    $\sqcap$ and~$\sqcup$ are constructed, it is proved by induction
    that every relevant function~$\oldphi$ has the property specified in the
    claim.  The induction is on the arity of~$\oldphi$.

    To prove the 
lemma,
    we use the claim with $\vfuns=\vfuns_0'$,
    $\mathcal{P} = {I'}$ and, for each valued constraint~$t$ of~$I'\!$,
    $\oldphi=f'_t$.
\end{proof}

\begin{lemma}\label{lem:lsm_easy}
    If $\nfuns\subseteq\funs{\bddQ}$ and $\toVal\nfuns$ has an STP/MJN
    multimorphism, then $\nCSP(\nfinfuns)$ is \LSM{}-easy for every
    finite $\nfinfuns\subset\nfuns$.
\end{lemma}
\begin{proof}
    Let $\nfinfuns$ be a finite subset of $\nfuns$.  To any
    instance $I_\#$ of $\nCSP(\nfinfuns)$ there corresponds an
    instance~$I=\toVal{I_\#}$ of $\VCSP(\vfuns_0)$, where
    $\vfuns_0=\toVal {\nfinfuns}$: for each weighted constraint~$t$,
    the function $F_t$ is mapped to $f_t=\toVal{F_t}$ while the scope
    $\sigma_t$ remains unchanged.  Using
    Lemma~\ref{cor:equivInstExtn}, we may construct an equivalent
    instance $I'$ of $\VCSP(\vfuns_0')$ on the domain
    $D_1\times\dots\times D_n$ and a multisorted multimorphism
    $\langle\sqcap,\sqcup\rangle$ of that instance, where each~$\sqcap_i$ 
 induces   
    a
    total order. $\langle \sqcap,\sqcup \rangle$ induces a multisorted multimorphism
    of each~$f_t$.

    We now construct an instance $I''$ over the Boolean domain that is
    equivalent to~$I'$ and hence to~$I$.  For each $i$, $1\leq i\leq
    n$, introduce a set of $|D_i|+1$ Boolean variables
    $V_i=\{z_{i,a}\mid a\in D_i^+\}$, where $D_i^+=D_i\cup\{\bot\}$.
    Extend the total order on~$D_i^+$ by placing~$\bot$ below all
    elements of~$D_i$.  Define a nested sequence of subsets of $D_i^+$
    by $U_{i,a}=\{b\in D_i^+\mid b<a\}$.  The idea is that 
the bijection~$\pi$ that establishes the equivalence between~$I'$ and~$I''$
maps
    each domain
    element $a\in D_i$  
to the sequence of Boolean values that
    assigns~$1$ to all variables in $U_{i,a}$, and $0$ to the others
in~$V_i$.
    Consider the constraint asserting that only these $|D_i|$
    particular assignments to~$V_i$ are allowed.  This constraint can
    be represented by
    the crisp cost function~$f$ that assigns $f(\vecx)=0$ to these
    assignments and $f(\vecx)=\infty$ to all others. Note that $F(x) =
    \exp(-f(x))$ is log-supermodular.

    Note that we can use the same relation for any pair of sets~$D_i$
    and~$D_j$ with $|D_i|=|D_j|$ --- if $D_i$ and $D_j$ have different
    total orders then the relation is applied to the variables in
    $D_i^+$ in a different order than to the variables in $D_j^+$.  If
    we add these crisp valued constraints then there is a natural
bijection~$\pi$
   between $D_1\times \dots\times D_n$ and feasible
    assignments to Boolean variables $V_1\cup\dots\cup V_n$.  The
    variable $z_{i,a}$ where $a$ is the smallest (respectively
    largest) element of $D_i^+$ always takes on the value~1
    (respectively~0), and so these variables are redundant.  However,
    their introduction simplifies the description of some
    constructions later in the proof.

    Consider a valued constraint in $I'$ of arity $k$ that imposes the
    function $f'\in \vfuns_0'$, and, without loss of generality,
    assume that its scope is the first $k$ variables $x_1,\dots,x_k$.
    Add a corresponding valued constraint $f''$ to~$I''$ with
    $f''\colon2^{V_1\cup\dots\cup V_k}\to\extR$ defined as follows,
    where for convenience we are viewing $f''$ as a function on
    subsets of $V_1\cup\dots\cup V_k$ rather than as a function of
    $|V_1|+\dots+|V_k|$ Boolean variables:
    \begin{equation*}
        f''(A)=\begin{cases}
            f'(a_1,\dots,a_k),&\text{if $A=U_{1,a_1}\cup\dots\cup U_{k,a_k}$ 
                                            for some $(a_1,\dots,a_k)$;}\\
            \infty,&\text{otherwise}.
        \end{cases}
    \end{equation*}
    We claim $f''$ is submodular, i.e., $f''(A\cap B)+f''(A\cup B)\leq
    f''(A)+f''(B)$.  If either $f''(A)=\infty$ or $f''(B)=\infty$
    there is nothing to prove.  So $A=U_{1,a_1}\cup\dots\cup
    U_{k,a_k}$ and $B=U_{1,b_1}\cup\dots\cup U_{k,b_k}$ for some
    $(a_1,\dots,a_k),(b_1,\dots,b_k)\in D_1\times\dots\times D_k$.
    Then
    \begin{align*}
        &f''(A\cap B)+f''(A\cup B)\\
        &=f''\big((U_{1,a_1}\cap U_{1,b_1})\cup\dots\cup (U_{k,a_k}\cap U_{k,b_k})\big)\\
        &\qquad\null+f''\big((U_{1,a_1}\cup U_{1,b_1})\cup\dots\cup (U_{k,a_k}\cup U_{k,b_k})\big)\\
        &=f''(U_{1,a_1\sqcap_1 b_1}\cup\dots\cup U_{k,a_k\sqcap_k b_k})\\
        &\qquad\null+f(U_{1,a_1\sqcup_1 b_1}\cup\dots\cup U_{k,a_k\sqcup_k b_k})\\
        &=f'(a_1\sqcap_1 b_1,\dots, a_k\sqcap_k b_k)+f'(a_1\sqcup_1 b_1,\dots,a_k\sqcup_k b_k)\\
        &\leq f'(a_1,\dots, a_k)+f'(b_1,\dots,b_k)\\
        &=f''(A)+f''(B).
    \end{align*}

    Now take stock.  We have an instance $I''$ of Boolean \VCSP{},
    which is equivalent to~$I'$ and hence to~$I$.  It has at most
    $n(|D|+1)$ Boolean variables and it has $n$ more valued
    constraints than~$I$.  The number of distinct valued
    constraints in~$\vfuns_0''$ is~$|\vfuns_0''|\leq|\vfuns'_0|+|D|$;
    note that these come from a fixed set of cost functions
    independent of the instance~$I$ and hence of $I_\#$ itself.
  
    Now map the \VCSP{} instance $I''$ back to \nCSP{} to yield
    an instance~$I_\#''$ over the Boolean domain in which every valued
    constraint comes from a certain fixed set of cost functions
    $\nfuns_0''\subset\LSM$.  Specifically, $I''=\toVal{I_\#''}$ and
    $\vfuns_0''=\toVal{\nfuns_0''}$.  Since $I''$ is equivalent
    to~$I$, there is a bijection between the non-zero terms of
    $Z(I_\#)$ and $Z(I_\#'')$ that preserves weights, and hence
    $Z(I_\#)=Z(I_\#'')$.
\end{proof}

Lemma~\ref{lem:lsm_easy} shows that, if $\toVal\nfuns$ has an STP/MJN multimorphism,
then $\nCSP(\nfinfuns)$ is \LSM{}-easy for every finite $\nfinfuns\subset\nfuns$.
Lemma~\ref{lem:bis_easy} below strengthens the result by showing
that $\nCSP(\nfinfuns)$ is \nBIS{}-easy.
The strengthening applies when the weight functions in $\nfuns$ have arity at most two.

In order to do the strengthening, we need to generalise the notion of a binary submodular
function to cover binary functions over larger domains. Let~$D$ and~$D'$ be ordered sets.
Following~\cite{new}, we say that a function $f:D\times D' \rightarrow \extR$
is \emph{submodular} if, for all $r,s\in D$ and all $r',s'\in D'$,
$$f(\min(r,r'),\min(s,s')) + f(\max(r,r'),\max(s,s')) \leq f(r,s) + f(r',s').$$
To apply this concept here, suppose that~$f$ is a function with domain $D_i\times D_j$.
Given orders on~$D_i$ and~$D_j$, let $D_i(\ell)$ and $D_j(\ell)$ denote the 
$\ell$'th element of~$D_i$ and~$D_j$, respectively.
Submodularity of~$f$ is equivalent to saying that
the  $|D_i|\times |D_j|$ matrix $M_f$ 
satisfies the Monge property~\cite{RudolfWoeginger}
where, as in Section~\ref{sec:balance}, 
$(M_f)_{k \ell} = f(D_i(k),D_j(\ell))$.
We 
extend Definition~\ref{def:fnLSM} by saying
that a function~$F: D_i\times D_j \to \bddQ$
is 
log-supermodular
(with respect to the given orders) if the function $\toVal F$ is 
sub-modular
(with respect to the same orders).

\begin{lemma}\label{lem:bis_easy}
If $\nfuns\subseteq\funs{\bddQ}$ is a weighted constraint
language whose weight functions have
arity at most two and $\toVal\nfuns$ has an STP/MJN
multimorphism, then $\nCSP(\nfinfuns)$ is \nBIS{}-easy for every
finite $\nfinfuns\subset\nfuns$.
\end{lemma}

\begin{proof}
Let $\nfinfuns$ be a finite subset of $\nfuns$.  We use exactly the same construction as in the previous lemma, but go further and    show that every weight function $F''$ appearing in instance    $I_\#''$ is expressible in terms of unary weight functions in    $\uns{\{0,1\}}$, and the binary weight function $\IMP$ defined by    $\IMP(0,0)=\IMP(0,1)=\IMP(1,1)=1$ and $\IMP(1,0)=0$.  Moreover,    unary weight functions in $\uns{\{0,1\}}$ (even those taking    irrational values) can be approximated sufficiently closely by    polynomial-sized pps-formulas using $\IMP$    
\cite[Lemma~13.1]{LSM}.    
This will complete the proof, since $\nCSP(\IMP)\APred \nBIS$ by \cite[Theorem~5]{wbool}.

The task then, is to show that every weight function~$F''$ 
in instance $I_\#''$ is expressible
in terms of unary weight functions in $\uns{\{0,1\}}$ and $\IMP$.
We do this by considering, in turn, the different 
types of weight functions arising in~$I_\#''$.
The $n$ relations (crisp cost functions) that were introduced in~$I''$ to    impose a total order on the variables in the sets $V_i$ are    clearly implementable in terms of~$\imp=\toVal{\IMP}$.  
   
Every  other  weight function $F''$ is associated with a
cost function~$f''$ in~$I''$ 
that is an implementation  over the Boolean domain of    
a  cost function~$f'$ from~$I'\!$.
Since $f'\in \vfuns_0$, it has arity at most~$2$.
Our goal 
is to show that the function
$F'(\vecx) = \toValsymb^{-1}(f'(\vecx))=\exp(-f'(\vecx))$ is expressible in terms of 
unary weight functions in $\uns{\{0,1\}}$ and $\IMP$. 
If $f'$ is unary, this is immediate, so suppose $f'$ is binary.

To fix the notation, suppose that $f'$ is a function $f'\colon D_i \times D_j \to \extR$.
We can assume without loss of generality that $D_i$ and $D_j$
are disjoint (otherwise, rename some elements). Also, $D_i$ and $D_j$ are ordered according
to the linear order induced by $\sqcap$. Since $\langle \sqcap, \sqcup \rangle$ induces
a multisorted multimorphism of~$f'$ (see Definition~\ref{def:genSTP}), the function~$f'$ is 
submodular
(with respect to this order).

Building on the work of Rudolf and    Woeginger~\cite{RudolfWoeginger},
Cohen, Cooper, Jeavons and Krokhin~\cite[Lemma 4.5]{new}
have shown that every 
binary submodular function
is expressible as a positive linear combination of
certain simple 
binary submodular functions.
Translated to our setting by applying $\toValsymb^{-1}$, this says that~$F'$ is expressible as a product of certain    
simple basis functions, namely
the binary functions
\begin{equation*}
         B^\alpha_{a,b}(x,y)=\begin{cases}
                        \alpha,&\text{if $x\geq a$ and $y\leq b$;} \\                    
                        1,&\text{otherwise},
                     \end{cases}
                     \end{equation*}
for all $(a,b)\in(D_i,D_j)$ 
(with a similar set of binary functions  defined by replacing $x\geq a$ and $y\leq b$ by $x\leq a$ and    $y\geq b$),
where $\alpha$ is an arbitrary    constant in the range $[0,1]$.

$B_{a,b}^0(x,y)$ may be implemented as $\IMP(z_{i,a^-},z_{j,b})$,
where $a^-$ is the element immediately below $a$ in the total order    on~$D_i$.
(To see this, note that the constraint rules out the possibility that $z_{i,a^-}=1$, which corresponds to $x\geq a$
together with $z_{j,b}=0$, which corresponds to $y\leq b$.)
Let
\begin{equation*}
        U_\beta(z)=\begin{cases}
                      \beta,&\text{if $z=0$};\\
                      1,&\text{if $z=1$}.
                  \end{cases}
    \end{equation*} 
Then for $\alpha>0$,
the basis function 
$B^\alpha_{a,b}(x,y)$ 
may be    implemented as 
$$\IMP(z_{i,a^-},w)\,\IMP(z_{j,b},w)\*    U_\alpha(z_{j,b})\,U_{1/\alpha-1}(w),$$ where $w$ is a new    variable.\end{proof}

To use Lemmas \ref{lem:lsm_easy} and~\ref{lem:bis_easy}, we need to
perform some scaling.  For any $k$-ary weight function in
$F\in\nfuns$, let $m_F = \max \{f(\vecx) \mid \vecx\in \dom^k\}$.  Let
\begin{equation*}
    \Lambda(F) = \begin{cases}
                     \ F/m_F &\text{if $m_F>1$} \\
                     \ F     &\text{otherwise}
                 \end{cases}
\end{equation*}
and let $\Lambda(\nfuns) = \{\Lambda(F) \mid F\in\nfuns\}$.  Note that
$\Lambda(F)$ always takes values in $\bddQ$ and that, since $\nfuns$ is
conservative, $\func_1(\dom,\bddQ)\subseteq \Lambda(\nfuns)$.

We return, once more, to the proof of Theorem~\ref{thm:main}.

\begin{theorem}
\label{thm:LSM-BIS-easy}
    Let $\nfuns$ be a weakly log-supermodular, conservative weighted
    constraint language taking values in $\posQ$.
    \begin{itemize}
    \item For any finite $\nfinfuns\subset\nfuns$, there is a finite
        $\nfinfuns'\subset\LSM$ such that $\nCSP(\nfinfuns)\APred
        \nCSP(\nfinfuns')$.
    \item If $\nfuns$ consists of functions of arity at most two, then
        $\nCSP(\nfinfuns)$ is \nBIS{}-easy for any finite $\nfinfuns
        \subset \nfuns$.
    \end{itemize}
\end{theorem}
\begin{proof}
    By Theorem~\ref{thm:wlsm_to_stp}, $\toVal{ 
    \Lambda
    (\nfuns)}$ has an
    STP/MJN multimorphism.  The result follows from Lemmas
    \ref{lem:lsm_easy} and~\ref{lem:bis_easy} and the fact that
    $\nCSP(\nfuns)\APred \nCSP(\Lambda(\nfuns))$.
\end{proof}

Theorem~\ref{thm:main}, our classification of the complexity of
approximating $\nCSP(\nfuns)$, now follows from Theorems
\ref{thm:BIShard-SATequiv}, \ref{thm:tractable}
and~\ref{thm:LSM-BIS-easy}.

\section{Algorithmic aspects}
\label{sec:algorithmic}

Finally, we consider the algorithmic aspects of the classification of
Theorem~\ref{thm:main}.  Intuitively, there is an algorithm that
determines the complexity of \nCSP{} with constraints from a finite
language $\nhfuns$ plus unary weights
because weak log-modularity is essentially equivalent to balance and
weak log-supermodularity is essentially equivalent to the existence of
a STP/MJN multimorphism.  
As we will show below, balance
and the existence of STP/MJN
multimorphisms depend only on certain finite parts of the
weighted constraint language so balance is decidable
by~\cite{NonNegExact} and the existence of STP/MJN multimorphisms can
be determined by brute force, or by using more sophisticated methods
from~\cite{KZ}.

We need to determine whether the infinite language
$\nhfuns\cup\uns{\dom}$ is balanced.  Fortunately, it suffices to
check whether $\nhfuns\cup\uns{\dom}'$ is balanced, where $\uns{\dom}'
= \func_1(\dom,\{1,2\})$, which is finite.  (Note that it is not
enough to test whether $\nhfuns$ is balanced; also, there is nothing
special about 1 and~2: any pair of distinct, positive rationals would
do. In fact, $|\uns{\dom}'|=2^{|\dom|}$ and there are sets of size $|D|$ which would suffice, but
we do not need this here.)

\begin{lemma}
\label{lem:balance-unary}
    Let $\nhfuns$ be a finite, weighted constraint language taking
    values in $\posQ$.  The following are equivalent:
    (1)~$\nhfuns\cup\uns{\dom}'$ is
    balanced; (2)~every finite subset of $\nhfuns\cup\uns{\dom}$ is
    balanced; and (3)~$\nhfuns\cup\uns{\dom}$ is balanced.
\end{lemma}
\begin{proof}
    (2) and~(3) are equivalent because any pps-formula contains only a
    finite number of atomic formulas.  (2) trivially implies~(1),
    since $\uns{\dom}'$ is finite.  It remains to show that (1)
    implies~(2) so, towards this goal, suppose that
    $\nhfuns\cup\uns{\dom}'$ is balanced.  We must show that every
    finite subset of $\nhfuns \cup\uns{\dom}$ is balanced.  Suppose
    that such a subset contains $r$ functions in $\uns{\dom}\setminus
    \nhfuns$.

    Let $\{F_1, \dots, F_r\}$ be unary functions such that
    $F_i(d) = a_{i,d}$ ($i\in\{1, \dots, r\}$, $d\in\dom$) and let
    $\nfinfuns = \nhfuns \cup \{F_1, \dots, F_r\}$.  We may consider
    the $a_{i,d}$ as formal variables and treat a function
    $G\in\fclone{\nfinfuns}$ with free variables $\vecx$ as a function
    of both $\vecx$ and the $a_{i,d}$.  We will show that, for any
    function $G$ and any interpretation of the $a_{i,d}$ (i.e., any
    instantiation of the function symbols $F_i$ as concrete functions
    $\dom\to\posQ$), the matrices associated with $G$ have
    block-rank~$1$, thus establishing that $\nfinfuns$ is balanced.

    So, consider any $G\in\fclone{\nfinfuns}$ with arity $n\geq 2$ and
    choose any $k$ with $1\leq k<n$.  We will show that the $\dom^k
    \times \dom^{n-k}$ matrix $M_G(\vecx,\vecy)$ has block-rank~$1$ for
    any value of the $a_{i,d}$.  By Lemma~\ref{lem:bool_facts}
    part~(\ref{pt:br1_2by2}), it suffices to show that every $2\times
    2$ submatrix induced by rows $\vecx,\vecx'$ and columns $\vecy,
    \vecy'$ has block-rank~$1$.  By Lemma~\ref{lem:bool_facts}
    part~(\ref{pt:binary}), this happens if, and only if, every such
    submatrix has rank~$1$ or at least two zero entries, which happens
    if, and only if, the multivariate polynomial
    \begin{equation*}
        p = G(\vecx,\vecy') G(\vecx'\!,\vecy) G(\vecx'\!,\vecy')
            \big[G(\vecx,\vecy)G(\vecx'\!,\vecy')
                 - G(\vecx'\!,\vecy)G(\vecx,\vecy')\big]
    \end{equation*}
    is zero for all values of $\vecx,\vecx'\!,\vecy,\vecy'$ and for
    all values of the $a_{i,d}$.  
    (Note that, if the submatrix defined by a pair of rows and columns 
    does not have block-rank~$1$ but 
   has exactly one zero, then only one of the four possible choices for $\vecx,\vecx'\!,\vecy,\vecy'$ will make $p$ non-zero.)    
   
    We now fix $\vecx, \vecx'\!, \vecy,
    \vecy'$ and consider $p$ as a function of just the $a_{i,d}$.
    Our goal is to show that (for every choice of $\vecx, \vecx'\!, \vecy,
    \vecy'$), $p$ is identically~$0$.

    Consider first the case where every $a_{i,d}$ is a power of two.  Here,
    every atomic formula $F_i(z)$ defines the same function as some
    product $U_1(z)\cdots U_\ell(z)$ of atomic formulas from
    $\uns{\dom}'$ so $G$ is equivalent to some function in
    $\fclone{\nhfuns \cup \uns{\dom}'}$.  But $\nhfuns\cup\uns{\dom}'$
    is balanced by assumption, so $p=0$ whenever every $a_{i,d}$ is a
    power of two.  Therefore, $p=0$ over a space that is a product of
    infinite sets.  It follows from the Schwartz--Zippel lemma 
    or from \cite[Theorem 1.2]{Alon} that
    the only polynomial with this property is
    the zero polynomial, so $p$ is the zero polynomial and
    $\nhfuns \cup \{F_1, \dots, F_r\}$ is
    balanced for any set $\{F_1, \dots, F_r\}$ of unary weights.
\end{proof}

\begin{theorem}
\label{thm:decidability}
    There is an algorithm that, given a finite, weighted constraint
    language $\nhfuns$ taking values in $\posQ$, correctly makes one
    of the following deductions, where $\nfuns = \nhfuns \cup
    \uns{\dom}$:
    \begin{enumerate}
    \item $\nCSP(\nfinfuns)$ is in \FP{} for every finite
        $\nfinfuns \subset \nfuns$;\label{deduc:FP}
    \item $\nCSP(\nfinfuns)$ is \LSM{}-easy for every
        finite $\nfinfuns \subset \nfuns$ and \nBIS{}-hard
        for some such~$\nfinfuns$;\label{deduc:LSM}
    \item $\nCSP(\nfinfuns)$ is \nBIS{}-easy for every
        finite $\nfinfuns \subset \nfuns$ and
        \nBIS{}-equivalent for some such~$\nfinfuns$;\label{deduc:BIS}
    \item $\nCSP(\nfinfuns)$ is \nSAT{}-easy for every
        finite $\nfinfuns \subset \nfuns$ and
        \nSAT{}-equivalent for some such~$\nfinfuns$.\label{deduc:SAT}
    \end{enumerate}
    If every function in $\nhfuns$ has arity at most~2, the output is
    not deduction~\ref{deduc:LSM}.
\end{theorem}
\begin{proof}
    We reduce the problem to determining whether $\nhfuns \cup
    \uns{\dom}'$ is balanced,
    whether $\toVal{\nhfuns}$ has an STP/MJN multimorphism and whether
    $\nhfuns$ contains only functions of arity at most~2.  Balance of
    finite languages is
    decidable~\cite{NonNegExact}.  An STP/MJN multimorphism consists
    of two operations $\dom^2\to\dom$ and three operations
    $\dom^3\to\dom$, which must have certain easily checked properties
    with respect to each of the functions in $\toVal{\nhfuns}$.  Thus,
    we can determine the existence of an STP/MJN multimorphism by
    brute force, checking each possible collection of five operations,
    or by using the methods of Kolmogorov and \Zivny{}~\cite{KZ}.  It
    is clearly decidable whether $\nhfuns$ contains a function of arity
    greater than~2.

    By Lemma~\ref{lem:balance-unary}, if $\nhfuns\cup\uns{\dom}'$ is
    balanced, then so
    is any finite $\nfinfuns \subset \nhfuns \cup \uns{D}$.  Therefore,
    by Lemma~\ref{lem:NonNegExact_dichotomy}, $\nCSP(\nfinfuns)$ can
    be solved exactly in \FP{} so we output deduction~\ref{deduc:FP}.
    From this point, we assume that $\nhfuns \cup\uns{\dom}'$ is not
    balanced.

    Since $\nhfuns\cup\uns{\dom}'$ is not balanced, nor is
    $\nhfuns\cup\uns{\dom}$
    (Lemma~\ref{lem:balance-unary}).  Therefore, $\nhfuns\cup\uns{\dom}$
    is not weakly log-modular (Lemma~\ref{lem:logmod_balanced}) so there
    is a finite $\nfinfuns \subset \nhfuns \cup \uns{\dom}$ such that
    $\nCSP(\nfinfuns)$ is \nBIS{}-hard
    (Theorem~\ref{thm:BIShard-SATequiv}).

    $\toVal{\Lambda(\nhfuns \cup \uns{\dom})}$ has an STP-MJN
    multimorphism if, and only if, $\toVal{\Lambda(\nhfuns)}$ does
    (Observation~\ref{obs:unarymultimorphism}), and
    $\toVal{\Lambda(\nhfuns)}$ is a finite language so we can determine
    whether it has an STP-MJN multimorphism by exhaustive search.  If
    $\toVal{\Lambda(\nhfuns \cup \uns{\dom})}$ has an STP-MJN
    multimorphism, then, for all finite $\nfinfuns\subset\Lambda(\nhfuns
     \cup \uns{\dom})$, $\nCSP(\nfinfuns)$ is \LSM{}-easy
    (Lemma~\ref{lem:lsm_easy}).  Since any function in
    ${\Lambda(\nhfuns\cup\uns{\dom})}$ is a scalar multiple of
    some function in ${\nhfuns\cup\uns{\dom}}$, $\nCSP(\nfinfuns)$
    is also \LSM{}-easy for all finite $\nfinfuns\subset \nhfuns \cup
    \uns{\dom}$.  We output deduction~\ref{deduc:LSM}, unless every
    function in $\nhfuns$ has arity at most~2, in which case
    $\nCSP(\nfinfuns)$ is \nBIS{}-easy for all finite $\nfinfuns
    \subset \nhfuns \cup \uns{\dom}$ (Lemma~\ref{lem:bis_easy}) and we
    output deduction~\ref{deduc:BIS}.

    On the other hand, if $\toVal{\Lambda(\nhfuns \cup \uns{\dom})}$
    has no STP-MJN multimorphism, then $\Lambda(\nhfuns \cup
    \uns{\dom})$ is not weakly log-supermodular
    (Theorem~\ref{thm:wlsm_to_stp}).  Because $\Lambda$ is just a
    rescaling,
    $\nhfuns\cup\uns{\dom}$ is also not weakly log-supermodular.
    Therefore, there is a finite $\nfinfuns\subset \nhfuns \cup
    \uns{\dom}$ such that $\nCSP(\nfinfuns)$ is \nSAT{}-equivalent
    (Theorem~\ref{thm:BIShard-SATequiv} again).  We output
    deduction~\ref{deduc:SAT}.
\end{proof}

\section{Acknowledgement}

We thank the referees for their useful suggestions, which improved the presentation of the paper.
 
\bibliographystyle{plain}

\bibliography{\jobname}
 
\end{document}